\pdfoutput=1
\documentclass[10pt,a4paper]{article}

\usepackage{amsmath,amsgen}
\usepackage{amstext,amssymb,amsfonts,latexsym}
\usepackage{theorem}
\usepackage{microtype}
\usepackage{graphicx}

 \newcommand{\bs}{\bigskip}
 \newcommand{\ms}{\medskip}
 \newcommand{\n}{\noindent}
 \newcommand{\s}{\smallskip}
 \newcommand{\hs}[1]{\hspace*{ #1 mm}}
 \newcommand{\vs}[1]{\vspace*{ #1 mm}}

 \newcommand{\setempty}{\varnothing}
 \newcommand{\real}{\mathbb{R}}
 \newcommand{\nat}{\mathbb{N}}
 
 \newcommand{\integer}{\mathbb{Z}}
 \newcommand{\rational}{\mathbb{Q}}
 \newcommand{\complex}{\mathbb{C}}

 \newcommand{\prob}{{\mathrm{Prob}}}

 \newcommand{\etalc}{\textrm{et al.}}

 \newcommand{\CC}{\mathcal{C}}

 \newcommand{\HH}{\mathcal{H}}
 
 \newcommand{\KK}{\mathcal{K}}
 \newcommand{\LL}{\mathcal{L}}

 \newcommand{\MM}{\mathcal{M}}
 
 \newcommand{\SSS}{\mathcal{S}}
 \newcommand{\PP}{\mathcal{P}}

 \newcommand{\UU}{\mathcal{U}}

 \newcommand{\dl}{\mathrm{L}}
 \newcommand{\nl}{\mathrm{NL}}
 \newcommand{\p}{\mathrm{P}}
 \newcommand{\np}{\mathrm{NP}}

\theoremstyle{plain}
\theoremheaderfont{\bfseries}
 \newtheorem{theorem}{Theorem}[section]
 \newtheorem{lemma}[theorem]{Lemma}
 \newtheorem{proposition}[theorem]{Proposition}
 \newtheorem{corollary}[theorem]{Corollary}

 {\theorembodyfont{\rmfamily}  \newtheorem{definition}[theorem]{Definition}}
 {\theorembodyfont{\rmfamily} \newtheorem{example}[theorem]{Example}}
 {\theorembodyfont{\rmfamily} }

 \newtheorem{claim}{Claim}

 \newenvironment{proof}{\par \noindent
            {\bf Proof. \hs{2}}}{\hfill$\Box$ \vspace*{3mm}}

 \newcommand{\ceilings}[1]{\lceil #1 \rceil}

 \newcommand{\qubit}[1]{| #1 \rangle}
 \newcommand{\bra}[1]{\langle #1 |}
 \newcommand{\ket}[1]{| #1 \rangle}
 \newcommand{\measure}[2]{\langle #1 | #2 \rangle}

\newcommand{\ignore}[1]{}

 \newcommand{\bql}{\mathrm{BQL}}

 \newcommand{\trace}{\mathrm{tr}}

 \newcommand{\bqlogtime}{\mathrm{BQLOGTIME}}
 \newcommand{\ac}[1]{\mathrm{AC}^{ #1 }}
 \newcommand{\nc}[1]{\mathrm{NC}^{ #1 }}
 
 \newcommand{\mmid}{\:\|\:}
 \newcommand{\logcfl}{\mathrm{LOGCFL}}
 \newcommand{\ilog}{\mathrm{ilog}}
 \newcommand{\iloglog}{\mathrm{iloglog}}
 \newcommand{\ptime}{\mathrm{ptime}\mbox{-}}

\begin{document}

\pagestyle{plain}
\pagenumbering{arabic}
\setcounter{page}{1}
\setcounter{footnote}{0}

\begin{center}
{\Large {\bf Quantum First-Order Logics That Capture Logarithmic-Time/Space Quantum Computability}}\footnote{An extended abstract  \cite{Yam24} appeared in the Proceedings of
the 20th Conference on Computability in Europe (CiE 2024), Amsterdam, the Netherlands, July 8-12, 2024, Lecture Notes in Computer Science, vol. 14773, pp.  311--323, Springer, 2024.} \bs\ms\\

{\sc Tomoyuki Yamakami}\footnote{Present Affiliation: Faculty of Engineering, University of Fukui, 3-9-1 Bunkyo, Fukui 910-8507, Japan}
\bs\\
\end{center}

\sloppy

\begin{abstract}
\sloppy
We introduce a quantum analogue of classical first-order logic (FO) and develop a theory of quantum first-order logic as a basis of the productive discussions on the power of logical expressiveness toward quantum computing.
The purpose of this work is to logically express ``quantum computation'' by introducing specially-featured quantum connectives and quantum quantifiers that quantify fixed-dimensional quantum states. Our approach is founded on the
recently introduced recursion-theoretical
schematic definitions of time-bounded quantum functions,
which map finite-dimensional Hilbert spaces to themselves.
The quantum first-order logic (QFO) in this work therefore looks quite different from the well-known old concept of quantum logic based on lattice theory.
We demonstrate that quantum first-order logics possess an ability of expressing bounded-error
quantum logarithmic-time computability by the use of new ``functional'' quantum variables. In contrast, an extra inclusion of quantum transitive closure operator helps us characterize quantum logarithmic-space computability.
The same computability can be achieved by the use of different ``functional'' quantum variables.

\s\n{keywords:} quantum computation, first-order logic, quantum Turing machine, unitary matrix, quantum quantifier, logarithmic-time/space  computability
\end{abstract}

\sloppy
\vs{-3}
\section{Background, Motivations, and Challenges}\label{sec:introduction}

We begin with brief descriptions of historical background, which motivates us to study a quantum analogue of the first-order logic and we then state the challenging tasks to tackle throughout this work.

\vs{-1}
\subsection{Quantum Computing and Quantum Logic}\label{sec:QC-and-QL}

A physical realization of quantum mechanical computing device has been sought for decades.
A theoretical framework of such quantum mechanical computing was formulated in the 1980s by Benioff
\cite{Ben80} and Deutsch \cite{Deu85}
as a quantum-mechanical  extension of classical computing.
Since then, theory of quantum computation has been well-developed to understand the power and limitation of quantum computing.
Fundamentally, quantum computing manipulates quantum states, which are  superpositions of classical states with (quantum) amplitudes, and it manages to conducts various types of quantum operations, which are simply quantum transformations of a finite number of \emph{quantum bits} (or \emph{qubits}, for short).

One of the simplest computational models used in theoretical computer science today is a family of Boolean circuits composed of AND, OR, and NOT gates.
Yao \cite{Yao93} made the first significant contribution to the development of \emph{quantum circuits}, each of which is built up from quantum gates, which work as unary transforms of quits.
A precursor to quantum circuits was nonetheless studied by Deutsch \cite{Deu89} under the name of quantum networks.

In Sections \ref{sec:character-QFO}--\ref{sec:functional}, nevertheless, we mostly focus our attention  on another machine model known as \emph{quantum Turing machine} (QTM), which is a quantum extension of classical Turing machine.
Following Deutsch's early model \cite{Deu85}.
QTMs were formulated by Bernstein and Vazirani \cite{BV97}. Subsequently, a multiple-tape variant of QTMs was studied in \cite{Yam99,Yam03}.
It was shown in \cite{Yao93} that uniform families of quantum circuits are equivalent in computational power to QTMs. The reader may refer to the textbooks, e.g., \cite{KSV02,NC00} for the references.

Quantum computation exploits a significant feature of quantum states, known as \emph{entanglement}, which is difficult to realize within a classical framework of logical terms.
Whenever we observe such a quantum state, it instantly collapses to a classical state with a certain probability.
This phenomenon irreversibly destroys the current state and it cannot be expressed in the classical framework.
A fundamental idea of this work, on the contrary, lies on an intention of how to ``translate'' quantum computing into ``logical expressions''.
Similar to time/space complexity, logical expressibility has been acknowledged as a useful complexity measure that provides us with a scalable value to express computational hardness of combinatorial problems.
Here, we attempt to take a more direct approach to ``expressing'' quantum computation described in the framework of \cite{Yam20,Yam22a}.

A long before the invention of quantum computing, an attempt to describe quantum physics in terms of logical expressions was made under the name of \emph{quantum logic} from a quite different perspective.
A large volume of work has been dedicated to ``express'' various aspects of quantum physics in numerous logical frameworks from an early introduction of quantum logic by Birkhoff and von Neumann \cite{BN36} to recent notions of quantum dynamic logic \cite{BS11}, quantum predicate logic \cite{Kor22}, quantum Hoare logic \cite{Yin09}, etc.
These modern quantum logics have been introduced to express various aspects of quantum physics.
In such a vast landscape of quantum logics, there is still an urgent need of logically expressing time/space-bounded quantum computing.
What kinds of quantum logics precisely capture such quantum computability?
This work intends to ``express'' various quantum complexity classes, which are naturally induced by QTMs in a quite different fashion.

As a natural deviation from standard conventions regarding ``inputs'', we intend to use quantum inputs and outputs in place of classical ones.
In most literature, the theory of quantum computation has coped with classical inputs and classical outputs except for a few special cases, in which quantum states are used as part of supplemental inputs during message exchanges in quantum interactive proof systems.
A more generic approach in this line of study was taken to deal with \emph{quantum NP} \cite{Yam02} and \emph{quantum functions}\footnote{This notion is different from the same name used in \cite{Yam03}.} \cite{Yam20} whose inputs are quantum states of finite-dimensional Hilbert spaces.

Another important aspect to mention is the use of quantifiers that quantify quantum states in a finite-dimensional Hilbert space. In the past literature, such quantifiers have been used as part of a requirement for the  ``acceptance'' of  machines in the setting of, for example, quantum interactive proofs and quantum $\np$.
Such quantifiers are generally referred to as \emph{quantum quantifiers} in comparison with classical quantifiers.
In \cite{Yam02}, a hierarchy over quantum NP was introduced by alternately applying quantum existential and universal quantifiers.

In this line of study, Yamakami \cite{Yam20} lately proposed a  recursion-theoretical ``schematic definition'' to capture  the notion of polynomial-time computable quantum functions.
Scheme-based constructions of quantum functions, which  are quite different from the machine- and circuit-based definitions, are more suited to measure the ``descriptional'' complexity of quantum functions.
The success of such a schematic description of quantum polynomial-time computing \cite{Yam20} and its followup \cite{Yam22b}  for quantum polylogarithmic-time computability motivates us to pursue the opportunity of further studying the ``expressibility'' of quantum computations and to investigate a quantum analogue of the first-order logic (FO), which we  intend to call QFO (Definition \ref{definition-QFO}).

\subsection{Quick Review on the Classical First-Order Logic or FO}\label{sec:schematic-def}

In computer science, it is of practical importance to clarify the ``complexity'' of given combinatorial problems. The standard complexity measures include  the execution time and the memory usage of Turing machines (as well as the circuit size of  families of Boolean circuits) necessary to solve the problems. The standard complexity classes, such as $\np$ and $\nl$ (nondeterministic log-space class), were introduced to reflect polynomial execution time and logarithmic space.


Classical logic has exhibited another aspect of computing. There is a long history of expressing classical time/space-bounded computability.
From a quite different perspective,
Fagin \cite{Fag74}, for example, gave a logical characterization of languages (or equivalently, decision problems) in $\np$.
Since then, numerous characterizations have been sought to capture the well-known complexity classes by way of expressing them using \emph{logical terms} and \emph{logical formulas}.
Theory of classical first-order logic has been developed in association with parallel computing, in particular, of families of Boolean circuits. The first-order logic is composed of important elements, including variables, predicate symbols, and function symbols, logical connectives, and logical quantifiers. Notice that Fagin's characterization of $\np$  requires the first-order logic together with  second-order existential quantifiers.
In mid 1980s,
Gurevich and Lewis \cite{GL84} explored a close connection between uniform $\ac{0}$ and the (classical) first-order logic (FO), and Immerman \cite{Imm87a,Imm87b} used FO to capture other complexity classes, such as $\p$ and $\nl$.

Immerman's characterizations of $\nl$ by first-order logic, for instance, includes three special binary predicate symbols $\{=,<,BIT\}$ over natural numbers, where $BIT(x,y)$ indicates that the $x$th bit of the binary representation of $y$ equals $1$ and it helps encode/decode natural numbers expressed  in binary within a logical system.
Barrington, Immerman, and Straubing \cite{BIS90} demonstrated that the (classical) first-order logic (FO) captures the family of all languages recognized by constant-depth alternating logtime Turing machines and thus precisely characterizes $\mathrm{DLOGTIME}$-uniform\footnote{Circuit complexity classes, such as $\ac{i}$ and $\nc{i}$, are generally known to be quite sensitive to the choice of the uniformity notions. Ruzzo \cite{Ruz81} discussed various types of uniformity notions.}  $\mathrm{AC}^0$.
Hereafter, $\mathrm{FO}$ denotes the family of all languages that are expressed as first-order sentences,  including $\{=,<,BIT\}$.
The use of $BIT$ is crucial because, without it, the first-order logic expresses only aperiodic regular languages \cite{MP71}.
Various extensions of the first-order logic have been proposed in the past literature.
Moreover, $\mathrm{FO}$ characterizes constant-time concurrent read/write parallel random-access machines allowing the shift operation \cite{Imm89}.

The first-order logic acts as a natural basis to the characterizations of various language families. For instance, the first-order logic equipped with  a transitive closure operator (TC) precisely characterizes $\nl$ \cite{Imm86,Imm89}. With an addition of a least fixed point operator (LFP), the first-order logic is powerful enough to capture $\p$ \cite{Imm82,Var82}.

Lindell \cite{Lin92} proved a general theorem that any uniformly extensible numerical predicate can be expressed by first-order formulas (including $\{=,<,BIT\}$).
He also remarked that $\mathrm{FO}$ is equivalent to the collection of languages expressed by arithmetical finite structures based on arithmetic operations of $+$ (addition), $\times$ (multiplication), and \hs{2}$\widehat{ }$\hs{1} (exponentiation).

Lautemann \etalc~\cite{LMSV01} further discussed \emph{groupoidal quantifiers} over groupoids and showed that $\logcfl$ is precisely characterized by such finite structures together with groupoidal quantifiers, where $\logcfl$ consists of languages that are log-space many-one reducible to context-free languages.
Other notable quantifiers include \emph{majority quantifier} and \emph{majority-of-pair quantifier} \cite{BIS90}
in connection to $\mathrm{TC}^0$.

\vs{-2}
\subsection{Our Challenges in This Work}\label{sec:challenge}

As noted earlier, the classical first-order logic (FO) has proven to be a quite useful means of expressing the complexity classes
(specifically, ``descriptive'' complexity)
of combinatorial decision problems (or equivalently, languages).
This usefulness drives us to expand the scope of the study on the relationships between machine's computability and logical expressibility.
One such directions is the logical expressiveness of quantum complexity classes.
It is therefore quite natural to look for a quantum analogue of FO, dubbed as \emph{QFO},
the first-order logic
to deal with the exotic nature of quantum computing.
As concrete examples, we wish to ``express'' $\mathrm{BQLOGTIME}$ (quantum logtime complexity class) and $\mathrm{BQL}$ (quantum logspace complexity class).
We wish to take a step along the great success of the study on FO and to pursue the same goal of capturing quantum computation.

In this work, we intend to study QFO in hopes of capturing time/space-bounded quantum computability, in particular, $\mathrm{BQLOGTIME}$ and $\mathrm{BQL}$.
Following a great success of the study on the classical first-order logic (FO) in computational complexity theory, as briefly explained in Section \ref{sec:schematic-def},
For this purpose, we wish to present a new formulation of QFO  by adopting a quantum schematic approach of \cite{Yam20,Yam22b} because it looks closer to a ``logical'' system than the existing machine models, such as QTMs and quantum circuits.
QFO looks distinctively away from the existing quantum logics in the literature.

We stress that our approach is completely different from the framework of the aforementioned quantum logic. We do not even attempt in this work to insist on any replacement of this  well-established old concept since our goal is to describe quantum computation and to analyze the expressing power of quantum first-order logical systems within the scope of computer science.

The most significance is the expressiveness of logical formulas with variables, connectives, and quantifiers in direct connection to machine's computations. This motivates us to look into a quantum analogue of such expressiveness of classical logic. This work attempts to make a crucial step toward a descriptional expressiveness of quantum computations on quantum states.

We aim at proposing a reasonable and coherent formulation of such a quantum analogue of FO in order to express the exotic nature of quantum mechanical computing. However, there lie numerous difficulties in introducing the proper notion of \emph{quantum first-order logic} (QFO).
Since our primary purpose of this work is to capture the quantum computability from a logical viewpoint (rather than ordinary machine-based approaches), our approach toward the definitions of ``logical terms'' and ``logical formulas'' presented in Section \ref{sec:syntax} is solely founded on a ``computational'' aspect of quantum physics.
For this purpose, we modify the syntax and semantics of the classical connectives AND, OR, and NOT according to quantum computing.
In Section \ref{sec:basics}, we will provide a foundation to the introduction of QFO.
We further consider a natural restriction of QFO, called classicQFO, obtained from QFO by limiting the use of specific types of quantum quantifiers (see Definition \ref{def-iqq}).

As noted in Section \ref{sec:QC-and-QL}, our approach in this work is essentially different from well-known ``quantum logic'' of Birkhoff and von Neumann although those concepts reflect the actual behaviors of quantum bits (or qubits).
A significant feature of our system of quantum first-order logic (QFO) comes from the fact that it is founded on the schematic definition of quantum computing \cite{Yam20,Yam22b}.

The rest of this work attempts to answer the following generic but fundamental questions concerning $\mathrm{QFO}$ and its variant $\mathrm{classicQFO}$.

\renewcommand{\labelitemi}{$\circ$}
\begin{itemize}\vs{-1}
  \setlength{\topsep}{-2mm}%
  \setlength{\itemsep}{1mm}%
  \setlength{\parskip}{0cm}%

\item[(1)] What are the expressing powers of $\mathrm{QFO}$ and $\mathrm{classicQFO}$?
\end{itemize}

The $\mathrm{TC}$-operator has proven to be a useful operator for extending $\mathrm{FO}$. In a similar fashion, we can consider a quantum analogue of TC, called $QTC$, and its restricted form, called $logQTC$.

\renewcommand{\labelitemi}{$\circ$}
\begin{itemize}\vs{-1}
  \setlength{\topsep}{-2mm}%
  \setlength{\itemsep}{1mm}%
  \setlength{\parskip}{0cm}%

 \item[(2)] What are the powers of $QTC$ and $logQTC$ when they are provided to the underlying logical systems $\mathrm{QFO}$ and $\mathrm{classicQFO}$?
\end{itemize}

Even a restricted form of second-order variables (and their quantifications) provides enormous power to the underlying logical system $\mathrm{FO}$. Similarly, we can consider a restricted form of second-order quantum variables to expand $\mathrm{QFO}$ and $\mathrm{classicQFO}$.

\renewcommand{\labelitemi}{$\circ$}
\begin{itemize}\vs{-1}
  \setlength{\topsep}{-2mm}%
  \setlength{\itemsep}{1mm}%
  \setlength{\parskip}{0cm}%

 \item[(3)] What are the powers of a restricted form of second-order quantum variables for the underlying logical systems $\mathrm{QFO}$ and $\mathrm{classicQFO}$?
\end{itemize}

{\n} The above three questions will be answered partially in Sections \ref{sec:character-QFO} and \ref{sec:functional}.
For the sake of the reader, they are illustrated in Figure \ref{fig:class-hierarchy}. The complexity classes $\mathrm{HBQLOGTIME}$ and $\mathrm{HBQL}$ respectively stand for the closures of $\mathrm{BQLOGTIME}$ and $\mathrm{BQL}$ under restricted quantum quantifications (see Section \ref{sec:logtime-QTM}). Moreover, the complexity class $\mathrm{DLOGTIME}$ consists of all decision problems solvable in $O(\log{n})$ time and $\mathrm{HLOGTIME}$ is the quantification closure of $\mathrm{DLOGTIME}$.

\begin{figure}[t]
\centering
\includegraphics*[height=5.0cm]{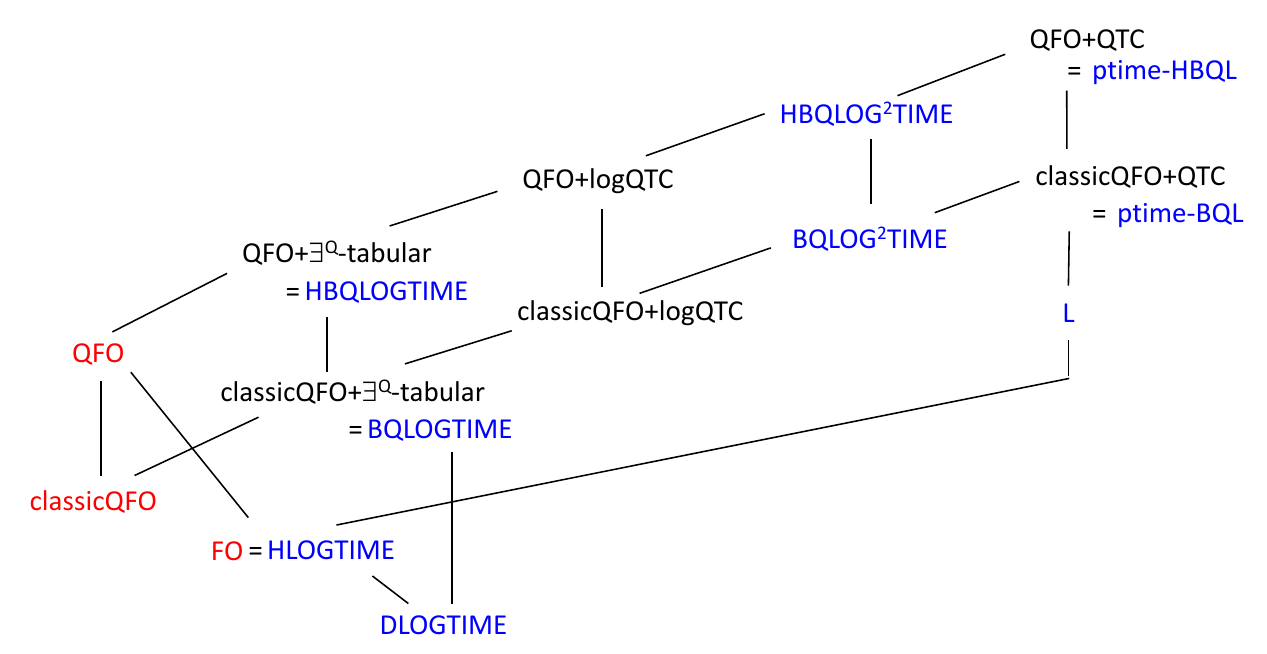}
\caption{Inclusion relationships among the complexity classes discussed in this work.  In \cite{BIS90}, HLOGTIME was expressed as HL.
}\label{fig:class-hierarchy}
\end{figure}

\vs{-2}
\section{Basic Definitions}\label{sec:basics}

Let us begin with providing fundamental notions and notation necessary to read through the rest of this work.

\vs{-2}
\subsection{Numbers, Languages, and Promise Problems}\label{sec:numbers}

To stand for the collection of \emph{natural numbers} including $0$, we use the notation $\nat$.
Moreover, the notation $\nat^{+}$ is meant for $\nat-\{0\}$.
Given integers $m$ and $n$ with $m\leq n$, the \emph{integer interval} $[m,n]_{\integer}$ refers to the set $\{m,m+1,m+2,\ldots,n\}$ and $[n]$ abbreviates $[1,n]_{\integer}$ when  $n\in\nat^{+}$.
The set of all \emph{complex numbers} is denoted $\complex$.
The symbol $\imath$ expresses $\sqrt{-1}$ and $e$ stands for the base of the natural logarithms.
All \emph{logarithms} in this work are taken to the base $2$ and, in particular, we define $\ilog(n)$ and $\iloglog(n)$ to be $\ceilings{\log{n}}$ for $n\geq1$ and $\ceilings{\log\log{n}}$ for $n\geq2$, respectively.
We further set $\hat{n}=2^{\ilog(n)}$, which satisfies that $n\leq \hat{n}<2n$.

A directed graph $G=(V,E)$ with $V=\{v_1,v_2,\ldots,v_n\}$ is said to be \emph{topologically ordered} if, for any $i,j\in[n]$, $(v_i,v_j)\in E$ and $i\neq j$ imply $i<j$.

Generally, a finite nonempty set of ``symbols'' or ``letters'' is referred to as an \emph{alphabet}.
We mostly use the binary alphabet $\{0,1\}$ and binary strings over it.
In particular, we use the notion $\lambda$ for the \emph{empty string}.
Notice that, for any string $u$ of length $n$, $u$ can be expressed as $u_{(1)}u_{(2)}\cdots u_{(n-1)}u_{(n)}$.
A \emph{language} over the binary alphabet is a subset of $\{0,1\}^*$. Given a number $n\in\nat$, $\{0,1\}^n$ expresses the set of all binary strings of length $n$, where the \emph{length} of a string $x$ is the total number of symbols in it and is denoted $|x|$.
Given a string $x$ of length $n$ and an integer $i\in[n]$, the notation $x_{(i)}$ denotes the $i$th symbol of $x$ from the left. By expanding this notation, whenever $i>n$ or $i=0$, we set $x_{(i)}$ to be $\lambda$.
We assume the standard \emph{lexicographic order} on $\{0,1\}^n$ for each $n\in\nat$ and we write  $s^{(n)}_i$ for  the $i$th string in $\{0,1\}^n$ in the lexicographic order; for example, $s^{(3)}_1=000$, $s^{(3)}_2=001$, $s^{(3)}_3=011$, $s^{(3)}_7=110$, and $s^{(3)}_8=111$.
Later, we freely identify $i$ with $s^{(\ilog(n))}_i$ as long as $n$ is clear from the context and $i$ is in $[\hat{n}]$.

A \emph{promise (decision) problem} $\LL$ over the binary alphabet $\{0,1\}$ is a pair $(A,B)$ of disjoint subsets of $\{0,1\}^*$. We interpret each element in $A$ (resp., $B$) as a positive (resp., negative) instance, which is supposed to be accepted (resp., rejected) by an underlying machine.
When $A\cup B=\{0,1\}^*$, we express $B$ as $\overline{A}$ and identify the promise problem $(A,\overline{A})$ with $A$, which means a \emph{language}. Hence, we freely treat languages as promise (decision) problems.

A \emph{classical Turing machine} (TM) has a read-only input tape, a constant number of rewritable work tapes, and a rewritable index tape, where an \emph{index tape} is used to access bits of a binary input string written on the input tape by way of writing down the address (or location)  of these bits onto the index tape. When input size is $n$, the index tape needs to use $\ilog(n)$ cells. The notations $\mathrm{DLOGTIME}$ and $\mathrm{HLOGTIME}$ refer to the collections of all languages recognized respectively by deterministic TMs and alternating TMs running in $O(\log{n})$ time.  See \cite{BIS90} for more information. In addition to them, we also introduce $\mathrm{PLOGTIME}$ by the use of logtime unbounded-error probabilistic TMs. When their transition probabilities are restricted to a particular set $K$ (such as $\rational$ or $\{0,1/2,1\}$), we write $\mathrm{PLOGTIMNE}_{K}$
to emphasize the use of $K$.

\subsection{Classical First-Order Logic}\label{sec:classical-FO}

We quickly review the basic terminology of the classical first-order logic used in, e.g., \cite{BIS90}.
The classical first-order logic originally uses natural numbers as target objects but they can be replaced or supplemented by any other objects, such as graph vertices and matrix entries. The binary relation $BIT(x,y)$ on those natural numbers expresses  the truth value of the following statement: the $x$th bit of the binary
representation of $y$ is $1$. The symbol $X$ is an atomic predicate symbol representing a binary input string in such a way that $X(i)$ is true iff the $i$th bit of the associated string is $1$.
The \emph{syntax} of the classical first-order logic is a collection of rules used to construct \emph{(classical) terms} and \emph{(classical) formulas} from predicate symbols ($=$, $<$, $BIT$, $X$), constant symbols ($1$, $n$), logical connectives ($\wedge$, $\vee$, $\neg$), variables ($x,y,z,\ldots$), and classical quantifiers ($\exists$, $\forall$). The \emph{first-order language} $\LL$ is the set of all formulas appropriately built up by the syntax.
A \emph{sentence} in $\LL$ is
a formula that has no free variables.
In certain literature (e.g., \cite{Imm87a,LMSV01}), the symbols in $\{=,<,BIT\}$ are entirely excluded from $\LL$. For our purpose, however, we include them as part of the first-order logic.

A \emph{(classical) structure} (or \emph{model}) over $\LL$ is a tuple composed of an input string for $X$, numeral objects for constant symbols, and relations for the predicate symbols.
Given a classical formula $\phi$ and a classical structure $\MM$, we write $\MM\models \phi$ to express that $\phi$ is true by evaluating variables and predicate symbols in $\phi$ by $\MM$. We say that a language $A$ is \emph{expressible} by a sentence $\phi$ if, for any string $x$, $x\in A$ iff $\MM_x\models \phi$, where $\MM_x$ is any structure in which the special predicate symbol $X$ is interpreted to the given input string $x$.

We write $\mathrm{FO}$ for the set of all languages expressible by sentences in $\LL$.

A class $\CC$ of languages is informally said to be \emph{logically definable} in a logical system $S$ if all languages in $\CC$ are expressible by appropriate sentences in $S$.

\vs{-2}
\subsection{Quantum States and Quantum Operations}

A vector in a finite-dimensional Hilbert space is expressed using the \emph{ket notation} $\qubit{\cdot}$ as in $\qubit{\phi}$, while an element of its dual space is denoted $\bra{\phi}$ using the \emph{bra notation}.
The notation  $\measure{\psi}{\phi}$ denotes the \emph{inner product} of two vectors $\qubit{\phi}$ and $\qubit{\psi}$ in the same Hilbert space.
The \emph{norm} of $\qubit{\phi}$, which is denoted $\|\qubit{\phi}\|$, is defined to be $\sqrt{\measure{\phi}{\phi}}$.
Let $\trace_i(\qubit{\phi})$ denote the quantum state obtained from $\qubit{\phi}$ by \emph{tracing out} all qubit locations of $\qubit{\phi}$ except for the $i$th qubit.

We use the notation $\HH_{2^n}$ to express a $2^n$-dimensional Hilbert space.
A \emph{qubit} is a unit-norm vector in a 2-dimensional Hilbert space and we use $\qubit{0}$ and $\qubit{1}$ to express the normalized basis of the Hilbert space $\HH_{2}$.
Any vector in this space is thus described as $\alpha\qubit{0}+\beta\qubit{1}$ with $\alpha,\beta\in\complex$. More generally, a quantum state of $n$ qubits is a vector in a $2^n$-dimensional Hilbert space and we write $\qubit{s}$ for a string $s\in\{0,1\}^n$ to express a normalized basis vector in $\HH_{2^n}$.
Notice that the set $\{\qubit{s}\mid s\in\{0,1\}^n\}$
is called the \emph{computational basis} of $\HH_{2^n}$ since this set
\emph{spans} the entire space $\HH_{2^n}$.
When $\qubit{\phi} = \sum_{s}\alpha_s\qubit{s}$ in $\HH_{2^n}$ has unit norm,  $\sum_{s}|\alpha_s|^2=1$ must hold.
We use the notation $\imath_{n}$ to denote the special quantum state of dimension $2^n$ with equal amplitudes, namely,  $2^{-n/2}\sum_{s\in\{0,1\}^n}\qubit{s}$.
A \emph{quantum string} (or a \emph{qustring}, for short) $\qubit{\phi}$ of $n$ qubits refers to a unit-norm quantum state of $n$ qubits and $n$ is the \emph{length} of $\qubit{\phi}$, expressed as $\ell(\qubit{\phi})$.
The notation $\Phi_{2^n}$ denotes the collection of all qustrings of length $n$. As a special case, we make the set $\Phi_0$ composed only of the null vector.
and we call the null vector by the qustring of length $0$.
Let $\Phi_{\infty} = \bigcup_{n\geq0} \Phi_{2^n}$. See, e.g., \cite{Yam99,Yam03,Yam20,Yam24} for more detail.

We expand the scope of promise problems $(A,B)$ from classical strings to qustrings in $\Phi_{\infty}$ by resetting the notation $(A,B)$ to satisfy $A,B\subseteq \Phi_{\infty}$ and $A\cap B=\setempty$.
Quantum NP, for instance, is defined to consist of promise problems over $2^n$-dimensional Hilbert space \cite{Yam02}.
As a concrete example, consider the promise problem $(A,B)$ defined by $A=\{\qubit{\phi}\in\Phi_{\infty}\mid \|\measure{0}{\psi}\|^2\geq 2/3\}$ and $B=\{\qubit{\phi}\in\Phi_{\infty}\mid \|\measure{1}{\psi}\|^2\geq 2/3\}$, where $\qubit{\psi} = \trace_{\ell(\qubit{\phi})}(\qubit{\phi})$.
Clearly, $(A,B)$ satisfies both $A,B\subseteq \Phi_{\infty}$ and $A\cap B=\setempty$.

In Section \ref{sec:functional}, we also consider functions mapping $[m_1]\times [m_2]\times \cdots \times [m_k]$ to $\Phi_{2^{m_0}}$ with $c,c',k\in\nat^{+}$, $\sum_{i=1}^{k}m_i = c\ilog{n}$, and $m_0=c'\ilog{n}$. Those functions are also viewed as sets of the form: $\{(i_1,i_2,\ldots,i_k,\qubit{\phi_{i_1,i_2,\ldots,i_k}})\mid : \forall j\in[k] [ i_j\in [m_j] ]\}$ with $\qubit{\phi_{i_1,i_2,\ldots,i_k}}\in \Phi_{2^{m_0}}$.

A \emph{quantum gate} of fan-in $n$ acts as a unitary transform of a quantum state of $n$ qubits to another quantum state of $n$ qubits. Such a transform can be expressed as a $2^n\times 2^n$ unitary matrix over $\complex$. We freely identify a quantum gate with its associated unitary matrix. Remember  that, for any quantum gate, its fan-in and fan-out should be equal.
Typical examples of quantum gates are $I$ (identity), $NOT$ (negation), $CNOT$ (controlled-NOT), $S$ (phase shift), $ROT_{\theta}$ (rotation at angle $\theta$), $T$ ($\pi/8$ gate), and $WH$ (Walsh-Hadamard transform).
Those quantum gates work for bits $a,b\in\{0,1\}$ as $NOT(\qubit{a})= \qubit{1-a}$, $CNOT(\qubit{a}\qubit{b}) = \qubit{a}\qubit{a\oplus b}$, $S\qubit{a}= \imath^{a}\qubit{a}$, $ROT_{\theta}(\qubit{a}) = \cos\theta\qubit{a}+(-1)^a\sin\theta\qubit{1-a}$, $T\qubit{a}=e^{\frac{\pi a\imath}{4}}\qubit{a}$, and  $WH(\qubit{a})= \frac{1}{\sqrt{2}}(\qubit{0}+(-1)^a\qubit{1})$. Note that $NOT=ROT_{\pi}$ and $WH=ROT_{\pi/4}$.
Moreover, there are two special types of gates: input gates and output gates.

As an auxiliary operator, we introduce $SWAP$, whose value $SWAP(i,\qubit{\phi})$ indicates the quantum state obtained by moving the $i$th qubit of a quantum state $\qubit{\phi_y}$ to the first; namely, $SWAP(i,\qubit{\phi}) = \sum_{u\in\{0,1\}^n}\alpha_u\qubit{u^{(i)}}$ if $\qubit{\phi}=\sum_{u\in\{0,1\}^n}\alpha_u\qubit{u}$, where $u^{(i)}=u_{(i)}u_{(2)}u_{(3)}\cdots u_{(i-1)}u_{(i+1)}\cdots u_{(n)}$.  Its inverse $SWAP^{-1}(i,\qubit{\phi})$ works as $SWAP^{-1}(i,\qubit{\phi}) = \sum_{u\in\{0,1\}^n}\alpha_u\qubit{\tilde{u}^{(i)}}$, where $\tilde{u}^{(i)} = u_{(2)}u_{(3)}\cdots u_{(i-1)} u_{(1)} u_{(i)}\cdots u_{(n)}$.

The number of input qubits is used as a basis parameter to describe the ``size'' and ``depth'' of a given quantum circuit. Given a quantum circuit, its \emph{size} is the total number of quantum gates and its \emph{depth} is the length of the longest path from any input gate to an output gate.

We designate one qubit as an output qubit. We assume that each quantum circuit uses only one measurement at the end of computation to obtain a classical output with a certain probability. Here, we measure the output qubit in the computational base, $\{\qubit{0},\qubit{1}\}$, after computation terminates.

The set $\{S,T,CNOT,WH\}$ of quantum gates is known to be \emph{universal} \cite{BBC+95} in the sense that any given quantum transform $U$ can be approximated to within a desired accuracy by an appropriate quantum circuit made up from a universal set of quantum gates.
It follows from the \emph{Solovay-Kitaev theorem} \cite{Kit97} (see also \cite{KSV02,NC00}) that, for any quantum gate $U$ on $n$ qubits and for any number $k\in\nat^{+}$, we can construct an $O(n\log^c(2^kn))$-size quantum circuit $C$ consisting only of the quantum gates in this universal set so that the norm $\|U(\qubit{\phi})-C(\qubit{\phi})\|$ is upper bounded by $2^{-k}$, where $c>0$ is an absolute constant.


We reserve the notations $\forall x$ and $\exists x$ for the quantifiers whose bound variable $x$ ranges over classical objects (such as strings, numbers, and graph vertices).
In contrast, the notations $\forall\qubit{\phi}$ and $\exists\qubit{\phi}$  express quantum quantifiers that quantify qustrings $\qubit{\phi}$.
We remark that the quantum quantifiers range over an uncountable domain in opposition to the classical quantifiers.
For example, the notation $(\forall \qubit{\phi}\in\Phi_{2^{\ilog(n)}}) P(\qubit{\phi})$ expresses that a statement $P(\qubit{\phi})$ holds for ``all qustrings $\qubit{\phi}$ in $\Phi_{2^{\ilog(n)}}$''; in comparison,  $(\exists \qubit{\phi}\in\Phi_{2^{\ilog(n)}}) P(\qubit{\phi})$ indicates that $P(\qubit{\phi})$ holds for ``some of the qustrings $\qubit{\phi}$ in $\Phi_{2^{\ilog(n)}}$''.

\vs{-1}
\section{Definition of Quantum First-Order Logic}\label{sec:definition-QFO}

We formally provide the definitions of quantum first-order logic (QFO) and its relevant notions by explaining our syntax and semantics with several  concrete examples.
We begin with basic concepts of logical terms and logical formulas based on structures (or models) that  represent specific quantum states and quantum transformations permitted by quantum mechanics. In contrast with (classical) first-order logic, we intend to use the terminology of ``quantum first-order'' in a clear reference to qustrings, functions from tuples of natural numbers to qustrings, etc. This clearly  contrasts ``quantum second order'', which refers to quantum functions mapping tuples of qustrings to qustrings, as in \cite{Yam20}.

\vs{-1}
\subsection{A Preliminary Discussion on Quantum First-Order Logic}\label{sec:early-discussion}

Before giving the formal definitions of quantum first-order logic in Sections \ref{sec:syntax} and \ref{sec:semantics}, we briefly discuss a fundamental idea behind the formulation of our quantum first-order logic.
There is an apparent similarity between ``classical computation'' and ``quantum computation''.  This similarity motivates us to formalize a quantum variant of  ``first-order logic'', with ``quantum terms'' and ``quantum formulas'', based on the behaviors of QTMs.
Despite the use of quantum states, we wish to evaluate quantum formulas to be either  true or false.

(1) To work with general qustrings of multiple qubits, we intend to express such qustrings symbolically in the form of ``variables''. Each variable represents an arbitrary quantum state of $\ilog(n)$ qubits, which may be partially or fully entangled. Those variables are referred to as \emph{(pure-state) quantum variables} to distinguish them from \emph{classical variables} that express natural numbers.
In classical first-order logic, the symbol $X$ is treated as an atomic predicate symbol to represent an input to the logical system.
We also use natural numbers to specify which qubits we apply quantum operations to.
In sharp contrast, we wish to use $X$ as a special quantum function expressing an input  qustring of length $n$ so that $X(i)$ ``represents'' the $i$th qubit of this particular qustring.

(2) Classical quantifiers in the classical first-order logic quantify variables ranging over natural numbers, whereas we intend to quantify ``quantum variables'' over all possible quantum states in a Hilbert space. The use of quantifications over quantum states has already been made in the past literature for quantum interactive proof systems and quantum $\np$. Quantum $\np$ of \cite{Yam02}, for instance, consists of promise problems $(A,B)$ over a $2^n$-dimensional Hilbert space for which there exist appropriate QTMs $M$ satisfying that, for any quantum state $\qubit{\phi}$ of $n$ qubits, (1) if $\qubit{\phi}$ is in $A$, then there exists a quantum state $\qubit{\xi}$ of $m$ qubits satisfying that $M$ accepts $(\qubit{\phi},\qubit{\xi})$ with high probability and (2)  if $\qubit{\phi}$ is in $B$, then $M$ rejects $(\qubit{\phi},\qubit{\xi})$ with high probability for any given quantum state $\qubit{\xi}$ of $m$ qubits, where $m$ is a positive polynomial in $n$. The above two conditions (1) and (2) do not exclude the possibility that, for certain quantum states $\qubit{\xi}$ of $m$ qubits, $M$ neither accepts nor rejects $(\qubit{\phi},\qubit{\xi})$ with high probability. This naturally leads us to an introduction of ``quantum quantifiers'', which quantify quantum variables.

(3) We remark that a major deviation from the classical first-order logic is the treatment of $BIT(i,Y)$ (cf. \cite{BIS90,Lin92}). In the classical setting, $BIT(i,Y)$ or its abbreviation $Y(i)$ is treated as an ``atomic formula'' indicating that the $i$th bit of $Y$  is $1$, whereas, in our setting, $QBIT$ is a ``function'' and thus $QBIT(i,y)$ is a ``term'' that expresses the value of the $i$th qubit of $y$.

(4) Unlike classical computation, the multiple use of the same variables are in general impossible due to the \emph{no-cloning theorem} of quantum mechanics. Hence, a quantum predicate needs to take two separate inputs $\qubit{\phi}$ and $\qubit{\psi}$.
Those qustrings  are differentiated as the \emph{first argument} and the \emph{second argument} of the quantum predicate.
We thus need to force a quantum variable to appear only at most once at the first argument place of a quantum predicate and at most once at the second argument place of another quantum predicate.

(5) To indicate the number of qubits and the location of particular qubits in a qustring, nevertheless, we also need ``classical'' variables (expressed in binary) as well.
Recall that, in the classical case, we can combine a finite number of classical strings in a series. To do the same with qustrings, we later use a function from natural numbers to qustrings.

(6) We intend to introduce a quantum analogue of logical formulas, called ``quantum formulas'', which are built from ``quantum terms'' and ``quantum predicate symbols'' by the limited use of logical connectives and quantifiers.
Any quantum predicate symbol, which is evaluated to be either true or false, is meant to describe a unitary property of given quantum states.
Such a property is expressed in this work as a set of input-output pairs of a fixed unitary matrix.
Recall that a universal set of quantum gates can approximate an arbitrary quantum gate to within arbitrary accuracy.
Since the set $\{I,S,T,CNOT,WH\}$ forms a universal set of quantum gates \cite{BBC+95}, we introduce those quantum gates, say, $G$ in the form of ``quantum predicate symbols'', which represent ``$G(\qubit{\phi})=\qubit{\psi}$''. For this reason, each quantum predicate symbol has the first section and the second section, which
respectively correspond
to an input and an output of a quantum gate.

(7) Typical classical logical connectives are $\wedge$ (conjunction), $\vee$ (disjunction), and $\neg$ (negation). In our formulation of quantum first-order logic, we interpret $\wedge$ to represent the successive execution of two separate independent computations. In contrast, the most troublesome classical connectives are $\neg$ and $\vee$.
The classical negation is supported by DeMorgan's law. Since this law does not hold in the quantum setting, we abandon the use of classical negation in our formalism.
Concerning the connective $\vee$, the expression $P\vee R$ means that we need to choose the correct formula between $P$ and $R$. This process is not a reversible computation. We thus replace $\vee$ by the branching scheme of \cite{Yam20}.

(8) We need to distinguish between two different types of existential quantifiers for quantum variables over quantum states. The first type deals with quantum variables that appear only in the first argument place of a certain quantum predicate symbol. This means that the corresponding quantum states are treated as inputs to underlying quantum computation. The second type deals with quantum variables that appear first in the second argument place of a certain predicate symbol and then appears in the first argument place of another quantum predicate symbol. This means that the corresponding quantum states are first generated from other quantum states and then as inputs to another quantum computation. In this case, the quantum states are completely and uniquely determined.

(9) To determine the final outcome of a quantum computation, we need to measure a designated qubit of the quantum state obtained after the computation. In a similar manner, we need to formulate an action of measurement on a designated qubit that appears in a given quantum formula.

\vs{-1}
\subsection{Syntax}\label{sec:syntax}

Our primary targets are quantum states of multiple qubits and various properties of these quantum states. We intend to develop a generic methodology to ``express'' those objects and their characteristic properties.
For clarity reason, we freely use additional parentheses ``$($'' and ``$)$'' wherever necessary although these parentheses are not part of the formal definitions of quantum terms and quantum formulas.

A \emph{vocabulary} (or an \emph{alphabet}) $T$ is a set of quantum predicate symbols of fixed arities, quantum function symbols of fixed arities, and constant symbols (which are actually quantum function symbols of arity $0$).
Terms are inductively defined from variables, constants, and function symbols, and formulas are constructed from quantum terms, quantum predicate symbols, logical connectives, and classical and quantum quantifiers.
To differentiate ``terms'' and ``formulas'' in our setting from the classical ones, we tend to use the phrases ``quantum terms'' and ``quantum formulas''.
Those concepts are explained below.
We begin with an explanation of quantum terms.

\s
\n{\bf [Variables]}
We use two types of variables: \emph{classical variables}, which are denoted by $i,j,k,\ldots$,
and \emph{pure-state quantum variables}, which are denoted by $x,y,z,\ldots$  
Each pure-state quantum variable $y$ is said to have \emph{qubit size} $|y|$, which indicates a positive integer. This $y$ refers to a qustring of length $|y|$.
When $k=|y|$, we briefly call $y$ a \emph{$k$-qubit quantum variable}.

To indicate the number of qubits and the location of particular qubits in a qustring, nonetheless, we need ``classical'' variables (expressed in binary) as well.
Recall that, in the classical case, we can combine a finite number of classical strings as a series. To do the same with qustrings, we later use a function from natural numbers to qustrings under the name of \emph{functional variables}.

\s
\n{\bf [Function symbols]}
A \emph{classical function symbol} of arity $1$ is the successor function $suc$, where
the expression
$suc(i)$ means $i+1$. We inductively define $suc^{(0)}(i)=i$, $suc^{(j+1)}(i)=suc(suc^{(j)}(i))$ for any $j\in\nat$.
As customary, for a constant $e\in\nat^{+}$, we informally write $i+e$ for $suc^{(e)}(i)$. Classical function symbols of arity $0$,
which are
customarily called \emph{(classical) constant symbols}, include $0$, $1$, $\ilog(n)$, and $n$, where $n$ refers to the input size.
\emph{Quantum function symbols} of arity $2$ are $QBIT$ and $\otimes$.
The symbol $\otimes$ is called the \emph{tensor product}; however, its role is to express a symbolic adjacency of two terms.
We intend to write $QBIT(i,y)$ and $s\otimes t$.

It is important to remark that, unlike the classical case, once we apply a quantum transform $C$ to the tensor product $\qubit{\phi}\otimes \qubit{\psi}$ of two pure quantum states $\qubit{\phi}$ and $\qubit{\psi}$, $C(\qubit{\phi}\otimes \qubit{\psi})$ may not be separated to two pure quantum states. From this fact, once we form $s\otimes t$ and apply a quantum predicate symbol $\PP_{F}$ as $\PP(s\otimes t:z)$, we may no longer split $z$ into two pure-state quantum terms.

The \emph{instance function symbol} $X$ is a special quantum function symbol of arity $1$, which is distinguished from all other symbols. Inputs to $X$ are assumed to be a qustring of length $\ilog(n)$.
Notice that, in the classical setting, $X(\cdot)$ is treated as a predicate symbol instead of a function symbol.

\begin{definition}[Quantum terms]
\emph{Quantum terms} are defined in the following inductive way.
(1) Classical variables $i,j,k,\ldots$ and pure-state quantum variables $x,y,z,\ldots$ are quantum terms.
(2) Classical function symbols of arity $0$ (i.e., classical constants), which are $0$, $1$, $n$, and $\ilog(n)$, are quantum terms.
(3) If $i$ is a classical term, then $suc(i)$ is also a classical term.
(4) Given a quantum variable $y$, $|y|$ is a classical term and $y$ is called a quantum term of qubit size $|y|$.
(5) If $i$ is a classical/quantum term and $y$ is a quantum term, then $QBIT(i,y)$ is a quantum term of qubit size $1$.
For readability, we abbreviate $QBIT(i,y)$ as $y[i]$, which is called the \emph{$t$th qubit} of $y$ or the \emph{$i$th component} of $y$.
(6) If $s$ and $t$ are quantum terms of qubit size $l$, then $s\otimes t$ is also a quantum term of qubit size $2l$.
(7) If $s$ is a quantum term of qubit size $\ilog(n)$, then $X(s)$ is a quantum term of qubit size $1$ and both $|X|$ and $|X(i)|$ are classical terms.
(8) Quantum terms are only obtained by finite applications of (1)--(7).
When a quantum term is constructed only from classical variables and classical functions, it is emphatically called a \emph{classical term}.
Given a quantum term $t$, the notation $Var(t)$ denotes the set of all pure-state quantum  variables that appear during the above construction process of $t$. For instance, if $t$ is the form $QBIT(s,y)$, then $Var(t)$ is $\{s,y\}$.
\end{definition}

We then explain the notion of quantum formulas.

\s
\n{\bf [Predicate symbols]}
A quantum predicate symbol is meant to describe a unitary property of given quantum states.
Such a property is expressed in this work as a set of input-output pairs whose transformation is made by an application of a fixed unitary matrix.
Recall that a universal set of quantum gates can approximate an arbitrary quantum gate to within arbitrary accuracy.
Since the set $\{I,S,T,CNOT,WH\}$ forms a universal set of quantum gates \cite{BBC+95}, we define these quantum gates, say, $G$ in the form of ``quantum predicate symbols'', which represent ``$G(\qubit{\phi})=\qubit{\psi}$''.
Thus, each quantum predicate symbol has the first section and the second section, which respectively correspond to an input and an output of a quantum gate.

We use two types of predicate symbols: \emph{classical predicate symbols} of arity $2$, which are $=$, $\leq$, and $<$, and \emph{quantum predicate symbols} of arity $1$, which include $\PP_{I}$ (identity) and $\PP_{ROT_{\theta}}$ (rotation at $x$-axis) for any $\theta\in\real$.
Notice that we do not include $\PP_{CNOT}$ (controlled NOT) because, as shown in Example \ref{example-formula}(1), it can be defined using $\PP_{I}$ and  $\PP_{ROT_{\pi}}$.
It is important to note that we cannot determine, in general, whether or not two quantum states are identical for sure. Thus, we do not use $s=t$ for non-classical terms $s$ and $t$. See, e.g., \cite{KNT08} for this issue.
Another quantum predicate symbol of arity $2$ is $\simeq_{\varepsilon}$ and is called the \emph{(quantum) measurement predicate}.

For any quantum predicate symbol $R$, we use the expression of the form $R(s:t)$ with a colon (:) used as a separator between the first and the second terms.
In particular, the terms $s$ and $t$ in $R(s:t)$ are respectively called the \emph{first argument} and the \emph{second argument} of $R$.
When a pure-state quantum variable $y$ or its component $y[i]$ appears in the first argument place of a quantum predicate, it is said to be \emph{processed}.
For readability, we write $s\simeq_{\varepsilon} b$ instead of formally writing $\simeq_{\varepsilon}(s:b)$ as a predicate for a constant $b\in\{0,1\}$.

\s
\n{\bf [Quantum logical connectives]}
Quantum logical connectives include $\wedge$ (quantum AND), $\|$ (quantum OR), and $\neg^Q$ (quantum NOT).
In our formulation, we interpret $\wedge$ to represent the successive execution of two separate independent ``computations''.
In contrast, the most troublesome classical connectives are $\neg$ and $\vee$.
Concerning the standard connective $\vee$, the expression $P\vee R$ indicates that we need to choose the correct formula between $P$ and $R$ but this process is not  ``reversible'' in general. We thus replace $\vee$
by the branching scheme of \cite{Yam20} with the new symbol ``$\|$''.
The classical negation is supported by de Morgan's law but, since this law does not hold in the quantum setting, we should abandon the use of the classical negation in our formalism.
The quantum negation $\neg^Q$, in contrast, plays a distinctive role in our formalism. Notice that the ``negation'' of the final outcome of quantum computation is made at the very end of the computation by exchanging between accepting inner states and rejecting inner states. The use of our $\neg^Q$ reflects this computational ``negation''.

\s
\n{\bf [Quantifiers]}
Classical quantifiers, $\forall$ and $\exists$, range over numbers in $[0,\ilog(n)]_{\integer}$. Quantum quantifiers include  $\forall^Q$ (\emph{quantum universal quantifier})
and $\exists^Q$
(\emph{quantum existential quantifier}) and both of them range over all qustrings of length $\ilog(n)$.
(i.e., all elements in $\Phi_{2^{\ilog(n)}}$).
It is important to note that we do not use $[0,n]_{\integer}$ for the range of the classical quantifiers in this current definition. Later, however, we will discuss how to remove this restriction.


\begin{definition}[Quantum formulas]\label{def-quantum-formula}
(I)
If $s$ and $t$ are classical terms, then $s=t$, $s\leq t$, and $s<t$ are  called \emph{classical formulas}, which are also \emph{atomic quantum formulas}.
If $s,t$ are quantum terms with $Var(s)\cap Var(t)=\setempty$, the expressions $\PP_{I}(s:t)$ and $\PP_{ROT_{\theta}}(s:t)$ for any $\theta\in\real$ are \emph{atomic quantum formulas}. If $s$ is a quantum term, $i$ is a classical term, and $b$ is either $0$ or $1$, then the expression $s[i]\simeq_{\varepsilon}b$ is an\emph{ atomic quantum formula}, where $\varepsilon$ is a real number in $[0,1]$.

(II)
Quantum formulas are defined in the following inductive way.
(1) Atomic quantum formulas are quantum formulas.
(2) Given two quantum formulas $R_1$ and $R_2$, the expression $R_1\wedge R_2$ is also a quantum formula.
(3) For a quantum variable $y$ and a quantum term $s$, if $y[s]$ does not appear in the second argument places of $R_1$ and $R_2$, then the expression  $(y[s])[R_1 \mmid R_2]$ is a quantum formula, where $y[s]$ is called an \emph{antecedent component}.
(4) If $R$ is a quantum formula,
then  $\neg^QR$ is also a quantum formula.
(5) If $s$ is of the form $|y|$ or any classical term built up symbols in $\{0,1,\ilog(n),suc\}$, then  $(\exists i\leq s)R(i)$ and  $(\forall i\leq s)R(i)$ are both quantum formulas.
(6) If $y$ is a pure-state quantum variable and $s$ is a quantum term, then $(\forall^Q y,|y|=s)R(y)$ and $(\exists^Qy,|y|=s)R(y)$ are both quantum formulas.
(7) Quantum formulas are obtained only by the applications of (1)--(6).
Given a quantum formula $R$, let $Var(R)$ denote the set of all quantum variables that  appear in $R$.
\end{definition}

As a simple example of quantum formulas, let us consider $\PP_{CNOT}(y:z)$, which is set to be $(|y|=|z|) \wedge \PP_{I}(y[1]:z[1])\wedge (z[1])[\PP_{I}(y[2]:z[2]) \mmid \PP_{ROT_{\pi}}(y[2]:z[2])] \wedge (\forall i,3\leq i\leq |y|)(\PP_{I}(y[i]:z[i]))$. In this way, we can express $CNOT$ (controlled NOT) in our quantum logical system using $\PP_{CNOT}$. Since $\PP_{CNOT}$ is ``definable'' from $(I,ROT_{\pi})$, there is no need to  include $\PP_{CNOT}$ as our predicate symbol.

A quantum formula $\phi$ is called a \emph{sentence} if there is no free variable in $\phi$.
Moreover, $\phi$ is said to be \emph{query-free} if it does not include the instance function symbol $X$.
Variables appearing within the scope of quantifiers are said to be \emph{bound} and other variables are \emph{free}.
We say that $\phi$ is \emph{measurement-free} if $\phi$ contains no subformulas of the form $t\simeq_{\varepsilon}0$ and $t\simeq_{\varepsilon}1$ for any quantum terms $t$.


Unfortunately, all quantum formulas in Definition \ref{def-quantum-formula}
that we have constructed so far
are not in a proper ``admissible'' form.
Therefore, we need the notion of ``well-formedness'', which describes this admissible form.
Let us take a quick look at a few simple cases, which should be avoided because of the violation of the rules of quantum physics.

(1) Unlike classical computation, the duplicate use of the same variable is in general impossible due to the \emph{no-cloning theorem} of quantum mechanics.
Hence, in general, we cannot keep a copy of the same quantum states.
For example, in the expression $\PP_{F}(x[i]:y[i]) \wedge \PP_{G}(x[i]:z[i])$, the component $x[i]$ appears twice in the first argument places of $\PP_{F}$ and $\PP_{G}$ as if they are identical copies.
Since we cannot copy $x[i]$ in the quantum setting,
Thus, we cannot use it for two different quantum predicates.
We then need to force a quantum variable to appear only at most once at the first argument place of a quantum predicate and at most once at the second argument place of another quantum predicate.
To describe this situation,
we wish to introduce a graph that expresses a transformation sequence of quantum variables.
Given a quantum formula $R$, we define a directed graph $G_R=(V,E)$, which is called the \emph{variable connection graph} of $R$, as follows. Let $V=Var(R)$ and let $E$ consist of all pairs $(y,z)$ for two quantum variables $y$ and $z$ such that there exists a quantum predicate symbol of the form $\PP_F$ in a given quantum formula for which either $y$ or its component $y[i]$ is the first argument of $R$ and either $z$ or
its component
$z[j]$ is the second argument of $R$.

(2) Since each quantum predicate realizes a quantum transform, its second argument is uniquely determined by its first argument. Therefore, the second argument place cannot contain constants or the instance function symbol $X$.

(3) Similarly to (2), the expression like $\PP_{F}(x[i]:z[i])\wedge \PP_{G}(y[i]:z[i])$ must be avoided because $z[i]$ could have two different values in $\PP_{F}$ and $\PP_{G}$ depending on the value of $x[i]$ and $y[i]$.

Any second argument of each predicate $\PP_{F}$ must be uniquely identified from its first argument. For any quantum subformula $\phi$ in $R$ of the form $\PP_{G}(y[t(i_1,\ldots,i_k)]: z[s(j_1,\ldots,j_l)])$ with quantum terms $t(i_1,\ldots,i_k)$ and $s(j_1,\ldots,j_l)$, we demand that, for any tuple $(i_1,\ldots,i_k)$, the term  $s(i_1,\ldots,i_k)$ over numbers is uniquely identified from $t(i_1,\ldots,i_k)$.

(4) Consider the expression $(\forall i\leq s)[\PP_{F}(x[i]:z[i])]\wedge (\forall j\leq t)[\PP_{G}(y[j]:z[j])]$, which does not fall into (1).
However, when the scopes of two classical quantifiers $(\forall i\leq s)$ and $(\forall  j\leq t)$ overlap, two distinct components $z[i]$ and $z[j]$ coincide as long as $i$ equals $j$.
We thus need to avoid such an expression.

Another avoidable expression is of the form $(\forall i\leq s)(\forall j\leq s)\PP_{F}(y[i]:z[j])$, in which there is a chance that $y[i_1]$ and $y[i_2]$ for two different values $i_1$ and $i_2$ correspond to the same $z[j]$; that is, $\PP_{F}(y[i_1]:z[j])$ and $\PP_{F}(y[i_2]:z[j])$ occur for two values $i_1$ and $i_2$. A similar situation occurs even in the case of $(\forall i\leq s)(\exists j\leq s)\PP_{F}(y[i]:z[j])$. All these situations should be avoided.

\begin{definition}[Well-formedness]
A quantum formula $R$ is called \emph{well formed} if it satisfies the following conditions (1)--(4).

(1) The variable connection graph of $R$ must be topologically ordered. In this case, we succinctly say that the quantum formula $R$ is \emph{order-consistent}.

(2) For any quantum predicate, its second argument place must contain no terms made up of constants or the instance function symbol $X$.

(3) Neither a pure-state quantum variable $z$ nor its component $z[i]$ appears in the first argument place (as well as the second argument places) of two different quantum predicates.

(4) For any pure-state quantum variable whose components are bound by classical quantifiers, no such component (with possibly ``different'' classical variables) appear in the first argument places (as well as the second argument places) of more than two quantum predicates.
\end{definition}

Since we are interested only in well-formed quantum formulas, for readability, hereafter, all quantum formulas
are implicitly assumed to be well-formed and we often drop the term ``well-formed'' from ``well-formed quantum formulas.''


\begin{example}\label{example-formula}
To promote the reader's understanding of quantum formulas,  we intend to present several simple examples of quantum formulas. In the first example, we assume that $|y|=|z|\geq2$.
In the third example, $|y|=|z|\geq3$ is assumed. The notation $i=j=k$ is a shorthand for $i=j\wedge j=k\wedge j=k$.
The special symbol ``$\equiv$'' is used to mean ``defined by''.

\renewcommand{\labelitemi}{$\circ$}
\begin{enumerate}\vs{-1}
  \setlength{\topsep}{-2mm}%
  \setlength{\itemsep}{1mm}%
  \setlength{\parskip}{0cm}%

\item $\PP_{copy}(y:z) \equiv (|y|=|z|)\wedge (\forall i\leq|y|) [ \PP_I(y[i]:z[i]) ]$.

\item $\PP_{NOT}(y:z) \equiv (|y|=|z|=1)\wedge \PP_{ROT_{\pi}}(y:z)$. (negation)

\item $\PP_{CNOT}(y:z) \equiv (|y|=|z|) \wedge \PP_{I}(y[1]:z[1])\wedge (z[1])[\PP_{I}(y[2]:z[2]) \mmid \PP_{NOT}(y[2]:z[2])] \wedge (\forall i,3\leq i\leq |y|)(\PP_{I}(y[i]:z[i]))$. In this way, we can express $CNOT$ (controlled NOT) in our quantum logical system. This is the reason why we have not included $\PP_{CNOT}$ as our basic predicate symbol.

\item $\PP_{CSWAP}(y:z) \equiv (|y|=|z|)\wedge \PP_{I}(y[1]:z[1])\wedge (z[1])[ \PP_{I}(y[2]:z[2]) \wedge \PP_{I}(y[3]:z[3]) \mmid \PP_{I}(y[2]:z[3])\wedge \PP_{I}(y[3]:z[2])] \wedge (\forall i,4\leq i\leq |y|)(\PP_{I}(y[i]:z[i]))$. (Controlled-SWAP)

\item $\PP_{WH}(y:z) \equiv (|y|=|z|)\wedge (\forall i\leq |y|) \PP_{ROT_{\pi/4}}(y[i]:z[i])$. (Walsh-Hadamard transform)

\item  $\PP_{\oplus}(y,z:v,w) \equiv (|y|=|z|=|w|)\wedge (\forall i\leq|y|) [\PP_{CNOT}(y[i]\otimes z[i]: v[i]\otimes w[i]) ]$. This $w$  classically expresses the bitwise XOR of $y$ and $z$.

\item $NEQ(y,z) \equiv (\exists^Q v,|v|=|y|)(\exists^Q w,|w|=|z|) [ \PP_{\oplus}(y,z:v,w)\wedge (\exists j\leq |w|)[w[j]\simeq_{\varepsilon}1] ]$. Classically, this expression indicates $y\neq z$.

  \item $ONE(y:z) \equiv (\exists i\leq |y|)[(y[i])[ (\forall k\leq |z|)\PP_I(1:z[k]) \mmid (\forall j<i) [ (y[j])[\PP_I(0:z[j]) \mmid \PP_I(1:z[j])] \wedge (\forall k,i<k\leq \ilog(n))[ (y[k])[\PP_I(0:z[k]) \mmid \PP_I(1:z[k]) ]]]$.

\item $MAJ_1(X)\equiv  (\exists^Qy,|y|=\ilog(n)) [ (\forall i\leq |y|) [\PP_{ROT_{\pi/4}}(0:y[i])] \wedge X(y)\simeq_{\varepsilon}1 ]$. This expression7 checks whether $X$ contains a large number of $1$s. Since the Walsh-Hadamard transform $WH$ works as $WH(\qubit{0^{\ilog(n)}}) = \frac{1}{\sqrt{\hat{n}}} \sum_{s\in\{0,1\}^{\ilog(n)}} \qubit{s} = \frac{1}{\sqrt{\hat{n}}} \sum_{i=1}^{\hat{n}} \qubit{s^{(\ilog(n))}_i}$, it thus follows that $\sum_{i=1}^{\hat{n}} \prob[\text{ $s^{(\ilog(n))}_i$-th bit of $x$ is $1$ }] = \#_1(x)/\hat{n}$, where $\#_b(x)$ denotes the total number of $b$ in $x$.
\end{enumerate}
\end{example}


Given a quantum formula $\phi$, a quantum variable $y$ (as well as its component $y[i]$) in $\phi$ is called \emph{predecessor-dependent} in $\phi$ if it appears in the second argument place of a certain quantum predicate symbol of $\phi$; otherwise, we call it \emph{predecessor-independent} in $\phi$.
For instance, in Example \ref{example-formula}(9), quantum variable $y$ and its component $y[i]$ in $MAJ_1(X)$  are both predecessor-dependent.
When all quantum variables used in $\phi$ are predecessor-dependent in $\phi$, we call this $\phi$ \emph{predecessor-dependent}.
For a practical viewpoint, it is useful to differentiate two usages of $\exists^Q$ in $\phi\equiv (\exists^Q y,|y|=s)R(y)$, depending on whether or not $y$ is predecessor-dependent in $\phi$.

\begin{definition}[Consequential/introductory quantum quantifiers]\label{def-iqq}
Given a quantum formula $\phi$ of the form $(\exists^Q y,|y|=s)R(y)$, if $y$ is predecessor-dependent in $\phi$, then $\exists^Qy$ is called \emph{consequential quantum existential quantifier}.
Otherwise, we call it an \emph{introductory quantum existential quantifier}. With this respect, quantum universal quantifiers are always  treated as introductory quantum quantifiers. The introductory quantum existential quantifiers together with quantum universal quantifiers are collectively called \emph{introductory quantum quantifiers}. We say that a quantum formula $\phi$ is \emph{introductory quantum quantifier-free} (or iqq-free, for short) if all quantum quantifiers used in $\phi$ are only consequential.
\end{definition}

Given a constant $\varepsilon_0\in[0,1/2)$, a quantum formula $\phi$ is said to be \emph{$\varepsilon_0$-error bounded} if all measurement predicate symbols appearing in $\phi$ have the form $t\simeq_{\varepsilon}0$ and $t\simeq_{\varepsilon'}1$ with $\varepsilon,\varepsilon'\leq \varepsilon_0$ for quantum terms $t$.

For the ease of description, we further expand quantum OR ($\|$) and quantum measurement ($\simeq_{\varepsilon}$) in the following way. Let $k$ denote any integer more than $1$. We inductively define $(z[i_1],z[i_2],\ldots,z[i_k])[P_1\mmid P_2\mmid \cdots \mmid P_{m_k}]$ to be $(z[i_{k}])[Q_1 \mmid Q_2]$, where $m_j= \sum_{i=1}^{j}2^i$ for $j\geq1$, $Q_1\equiv (z[i_2],z[i_3],\ldots,z[i_{k}])[P_1\mmid P_2 \mmid \cdots \mmid P_{m_{k-1}}]$, and $Q_2\equiv (z[i_2],z[i_3],\ldots,z[i_{k}])[P_{m_{k-1}+1}\mmid P_{m_{k-1}+2}  \mmid \cdots \mmid P_{m_k}]$.
We also define a \emph{multiple-qubit measurement} as $z[i_1]\otimes z[i_2] \otimes \ldots \otimes z[i_k]  \simeq_{\varepsilon} b_1b_2\cdots b_k$.

\vs{-1}
\subsection{Semantics}\label{sec:semantics}

Quantum formulas that we have introduced in Section \ref{sec:syntax} are made up of the following components: clasical/quantum terms, classical/quantum predicate symbols,  classical/quantum quantifiers, and quantum logical connectives.
Next,  we wish to provide semantics to these quantum formulas by giving intended ``meanings'' of them by an evaluation function based on a given structure (or a model).
This is intended to  ``realize'' each quantum computation of an underlying quantum machine in terms of logic as in the classical case of FO. Given an input $\qubit{\phi_X}$ to an underlying machine, the machine produces an address $i$ of the input and tries to access the $i$th bit of $\qubit{\phi_X}$.

\begin{definition}[Structure]
Given a vocabulary $T$, a \emph{structure} (or a \emph{model}) $\MM$ with vocabulary $T$ is a set  consisting of  a universe $\UU$, relations of fixed arities, functions of fixed arities, and constants from $\UU$. Each relation, each function, and each constant respectively correspond to their associated predicate symbol, functions symbol, and constant symbol.
\end{definition}

In what follows, we fix a structure $\MM$  arbitrarily.
We then  introduce the notion of interpretation, which assigns ``objects'' in the given structure $\MM$ to classical/quantum terms. Generally, we assign natural numbers to classical terms and quantum states to quantum terms.
To the instance function symbol $X$, in particular, we assign a qustring, say, $\qubit{\phi_X}$ of length $n$ given by $\MM$.

A pure-state quantum variable $y$ represents a qustring  $\rho_y=\qubit{\phi_y}$ of length $|y|$.
A quantum term $QBIT(i,y)$ for a classical term $i$ is intuitively meant to be the result of accessing the $\hat{i}$th qubit of $\rho_y$, where $\hat{i}$ is the number assigned to $i$. However, we cannot separate it from $\rho_y$, and thus we instead try to move it to the first location of $\rho_y$ by not destroying the quantum state $\rho_y$.
More precisely, for a classical term $i$ expressing $\hat{i}$, $QBIT(i,y)$ indicates  $SWAP(\hat{i},\qubit{\phi_y})$.
For a quantum term $s$ expressing  $\rho_s=\sum_{j\in[2^{|s|}]}\beta_j\qubit{j}$, in contrast, $QBIT(s,y)$ indicates $SWAP(\rho_s,\qubit{\phi_y}) = \sum_{j\in[2^{|s|}]} \beta_j SWAP(j,\qubit{\phi_y})$ (by the linearity of $SWAP$).

\begin{definition}[Interpretation]
An \emph{interpretation} $\xi$ is a function that assigns numbers and quantum states to classical and quantum terms, respectively. This $\xi$ assigns  natural numbers to classical variables $i$ and classical terms $suc(i)$ as  $\xi(i) = \hat{\imath}\in\nat$ and  $\xi(suc(i)) = \hat{\imath}+1$.
To a free pure-state quantum variable $y$, $\xi$ assigns  an arbitrary qustring $\rho_y=\qubit{\phi_y}$ of length $|y|$, namely,  $\xi(y) = \rho_{y}$.
Moreover, we assign to $X$ a qustring $\qubit{\phi_{X}}$ of length $n$, which is given in $\MM$.
For the quantum function symbols  $\otimes$ and $QBIT$, we use the following interpretation. The term $QBIT(i,s)$ is interpreted to
$\xi(QBIT(s,y)) = \xi(y[s]) = SWAP(\rho_{s},\rho_{y})$.
The term of the form $y[s]\otimes z[t]$ is interpreted to
$\xi(y[s]\otimes z[t]) = SWAP(2,SWAP(|s|,\rho_{y,s}\otimes \rho_{z,t}))$, where $\rho_{y,s} =\xi(y[s])$ and $\rho_{z,t}=\xi(z[t])$.
Moreover, $X(i)$ and $X(s)$ are interpreted to $\xi(X(i)) = SWAP(\hat{i},\qubit{\phi_X})$ and
$\xi(X(s)) = SWAP(\rho_s,\qubit{\phi_X})$, respectively.
\end{definition}

We treat quantum negation $\neg^Q$ in a quite different way that, unlike the classical case, $\neg^Q$ affects only quantum quantifiers associated with predecessor-independent quantum variables. This quantum negation comes from the asymmetric way of defining the acceptance and rejection criteria for Quantum $\np$ \cite{Yam02}.

\begin{definition}[Evaluation]\label{evaluation}
Let us define an evaluation $eval[R](\MM)$ of a quantum formula $R$ by $\MM$.

(1)
Consider an evaluation of predicate symbols.
Let $s,t,u,v$ denote four quantum terms.
For any transform $C\in\{I,ROT_{\theta}\mid \theta\in\real\}$ and for any two quantum states $\qubit{\phi_s},\qubit{\phi_t}\in \HH_2$, we define the value of
$eval[\PP_{C}(s:t)](\rho_s,\rho_t,\MM)$ as follows.

(i) In the case where $s$ and $t$ are pure-state variables of qubit size $1$, since $\rho_s=\qubit{\phi_s}$ and $\rho_t=\qubit{\phi_t}$, $eval[\PP_{C}(s:t)](\rho_s,\rho_t,\MM)$ is  $1$ if $\rho_{t} = C\rho_s$. Otherwise, it is $0$.

(ii) In the case where $s$ and $t$ are of the form $y[i]$ and $z[j]$ for classical terms $i$ and $j$ and pure-state variables $y$ and $z$ with $|y|=|z|$, respectively, we set $eval[\PP_{C}(s:t)](\rho_s,\rho_t,\MM)$ to be $1$ if $SWAP(\hat{j},\rho_z) = C\cdot SWAP(\hat{i},\rho_y)$ holds.
Otherwise, we set it to be $0$.

(iii) In the case where $s$ and $t$ are of the form $y[s']$ and $z[t']$ for pure-state variables $y$ and $z$ with $|y|=|z|$, $eval[\PP_{C}(s:t)](\rho_s,\rho_t,\MM) = 1$ if $SWAP(\rho_{t'},\rho_z) =C\cdot SWAP(\rho_{s'},\rho_y)$ holds. Otherwise, we set it to be $0$.

(2) For a classical formula of the form $s=t$, we set $eval[s=t](\hat{s},\hat{t},\MM)=1$ iff the classical numbers $\hat{s}$ and $\hat{t}$ interpreted for $s$ and $t$ are equal. The cases of $s\leq t$ and $s<t$ are similarly treated.
For any bit $b\in\{0,1\}$, we define $eval[s[i]\simeq_{\varepsilon}b](\qubit{\phi_s},b,\MM)$ to be $1$ if  where $s$ is interpreted to $\rho_{s[i]} = \sum_{c}\sum_{u:u_{(i)}=c}\alpha_u\qubit{u}$  and
$b$ is observed in the computational basis by quantum measurement with error probability at most $\varepsilon$, i.e.,  $\|\rho_s\|^2 - \|(I^i\otimes \ket{b}\bra{b} \otimes I^{|s|-i-1}) \rho_{s}\|^2\leq \varepsilon$. Otherwise, it is set to be $0$.

(3) For
the quantum AND ($\wedge$), we set $eval[P\wedge Q](\MM_{P\wedge Q})$ to be $1$ if both $eval[P](\MM_{P})=1$ and $eval[Q](\MM_{Q})=1$ hold for appropriate structures $\MM_P$ and $\MM_Q$ and $\MM_{P\wedge Q}$ is the union of $\MM_{P}$ and $\MM_{Q}$. Otherwise, we set it to be $0$.
For the quantum OR ($\mmid$), we set $eval[(y[s])[P_1\mmid P_2](\rho_s,\rho_y,\MM)$ to be $1$ if  $\rho_y$ is a quantum state interpreted for $y$ and both  $eval[P_1](\rho_s,\rho_y,\rho_{s,y}(0),\MM)=1$ and $eval[P_2](\rho_s,\rho_y,\rho_{s,y}(1),\MM)=1$ hold, where  $\rho_{s,y}(b) = \sum_{j\in[2^{|s|}]} \beta_j \sum_{u\in\{0,1\}^{|y|},u_{(j)}=b} \alpha_u\qubit{u}$ for any $b\in\{0,1\}$, $\rho_s=\sum_{j\in[2^{|s|}]}\beta_j\qubit{j}$, and $\rho_y=\sum_{u\in\{0,1\}^{|y|}}\alpha_u\qubit{u}$. Note that $\rho_y = \rho_{s,y}(0)+ \rho_{s,y}(1)$ holds.  Otherwise, we set $eval[(y[s])[P_1\mmid P_2](\rho_s,\rho_y,\MM)$ to be $0$.

(4) For quantifiers, we then consider then separately according to their types.

(i) For classical quantifiers, $eval[(\forall i\leq t)R(i)](\MM)=1$ iff $eval[R(i)](\hat{i}, \MM)=1$ and $eval[i\leq t](\hat{i},\MM)=1$ for all $\hat{i}\leq \hat{t}$, where $i$ and $t$ are interpreted to $\hat{i}$ and $\hat{t}$, respectively.
Similarly,  $eval[(\exists i\leq t)R(i)](\MM)=1$ iff $eval[R(i)](\hat{i},\MM)=1$ and $eval[i\leq t](\hat{i},\MM)=1$ for a certain $\hat{i}\leq \hat{t}$.

(ii) Concerning quantum quantifiers, on the contrary, we set $eval[(\exists^Qy, |y|=t)R(y)](\MM)$ to be $1$ if there exists a qustring $\qubit{\phi}$  of length $\hat{t}$ such that $eval[R(y)](\qubit{\phi_y},\MM)=1$  and $eval[|y|=t](\MM)=1$, provided that $\hat{t}$ is a number interpreted for $t$. Otherwise, we set it to be $0$.
Moreover, we set $eval[(\forall^Qy,|y|=t)R(y)](\MM)$ to be $1$ if $eval[R(y)](\qubit{\phi_y},\MM)=1$ and $eval[|y|=t](\MM)=1$ for all qustrings $\qubit{\phi_y}$ of length $\hat{t}$. Otherwise, it is $0$.

(5)
We define the value of $eval[(\neg^Q) P](\MM)$ inductively  for any quantum formula $P$ as follows.

(i) For any quantum predicate symbol of the form $\PP_{G}$ for $G\in\{I,ROT_{\theta}\mid \theta\in\real\}$,  $eval[(\neg^Q)\PP_{G}(s:t)](\MM)$ equals $eval[\PP_{G}(s:t)](\MM)$.
In contrast, we set $eval[(\neg^Q)(s\simeq_{\varepsilon}b)](\MM) = eval[s\simeq_{\varepsilon}\bar{b}](\MM)$, where $b$ is in $\{0,1\}$ and $\bar{b}=1-b$. On the contrary, when $s$ and $t$ are classical terms, we set  $eval[(\neg^Q)(s=t)](\MM) =  eval[s=t](\MM)$.

(ii) The double quantum negation makes $eval[(\neg^Q)(\neg^Q)P](\MM) = eval[P](\MM)$.
Moreover, $eval[(\neg^Q)(P\wedge Q)](\MM) = eval[(\neg^Q)P\wedge (\neg^Q)Q](\MM)$ and $eval[(\neg^Q)(y[s])[P_1\mmid P_2]](\MM) = eval[(y[s])[(\neg^Q)P_1\mmid (\neg^Q)P_2]](\MM)$.

(iii) If $y$ is a predecessor-independent variable, then we set $eval[(\neg^Q)(\exists^Q y,|y|=t)P(y)](\MM)$ to be $eval[(\forall^Q y,|y|=t)(\neg^Q)P(y)](\MM)$ and $eval[(\neg^Q)(\forall^Q y,|y|=t)P(y)](\MM)$ to be $eval[(\exists^Q y,|y|=t)(\neg^Q)P(y)](\MM)$.
On the contrary, for a predecessor-dependent variable $y$,  we set $eval[(\neg^Q)(\exists^Q y,|y|=t)P(y)](\MM)$ to be $eval[({\exists}^Q y,|y|=t)(\neg^Q)P(y)](\MM)$.
This is because the ``value'' of $y$ is completely and uniquely determined by a certain predicate that contains $y$ in its second argument place.

(iv)
For the classical universal quantifier ($\forall$),
we set $eval[(\neg^Q)(\forall i\leq t)P(i)](\MM)$ to equal $eval[(\forall i\leq t)(\neg^Q)P(i)](\MM)$.
\end{definition}

In the above definition of evaluation, not all quantum states that appear in the structure $\MM$ have unit norm.

Given two quantum formulas $\phi_1$ and $\phi_2$ including the same set of variables, we say that $\phi_1$ and $\phi_2$ are \emph{semantically equivalent} if, for any structure $\MM$ for $\phi_1$ and $\phi_2$,  it holds that $eval[\phi_1](\MM)=1$ iff $eval[\phi_2](\MM)=1$.
For instance, by Definition \ref{evaluation}(5), $\phi_1\equiv (\neg^Q)(y[s])[P_1\mmid P_2]$ is semantically equivalent to $\phi_2\equiv (y[s])[(\neg^Q)P_1\mmid (\neg^Q)P_2]$. Similarly, $(\neg^Q)(s\simeq_{\varepsilon}b)$ is semantically equivalent to $s\simeq_{\varepsilon}\bar{b}$.

Next, we introduce the notion of satisfiability.

\begin{definition}[Satisfiability, semantical validity]
A structure $\MM$ is said to \emph{satisfy} a quantum formula $\phi$ (or $\phi$ is \emph{satisfiable} by $\MM$)  if $\phi$ is evaluated to be $1$ by $\MM$ (i.e., $eval[\phi](\MM)=1$).  We use the notation $\MM\models \phi$  to express that $\phi$ is satisfiable by $\MM$.
A quantum formula $P$ is said to be \emph{semantically valid} (or simply \emph{valid}) if it is satisfiable by all structures.
\end{definition}

It is important to remark that, for any quantum sentence $\phi$ and for any structure $\MM$, there is no guarantee that either $\MM\models \phi$ or $\MM\models \neg^Q \phi$ always holds, and therefore there may be a chance for which both $\MM\not\models \phi$ and $\MM\not\models \neg^Q \phi$ hold.

\vs{-1}
\subsection{Basic Properties}

In classical logic, there are numerous rules that help us simplify complicated logical formulas. For instance, de Morgan's law helps transform  $x\wedge (y\vee z)$ to its logically equivalent formula $(x\wedge y) \vee (x\wedge z)$.


\begin{lemma}\label{logical-AND}
(1) If $(P\wedge Q) \wedge R$ and $P\wedge (Q\wedge R)$
are well-formed, then
they are  semantically equivalent to each other.
(2) If $P\wedge (x[1])[Q_1\mmid Q_2]$ and $(x[1])[P\wedge Q_1\mmid P\wedge Q_2]$, then
they are semantically equivalent to each other.
\end{lemma}


\begin{proof}
(1) This immediately follows from the fact that $eval[(P\wedge Q)\wedge R](\MM)=1$ is decomposed into $eval[P](\MM)=1$, $eval[Q](\MM)=1$, and $eval[R](\MM)=1$.

(2) This comes from the following facts: (i) $eval[P\wedge (x[1])[Q_1\mmid Q_2]](\MM,\rho_x)=1$ iff $eval[P](\MM)=1$ and $eval[(x[1])[Q_1\mmid Q_2](\MM,\rho_x)=1$, (ii) $eval[(x[1])[Q_1\mmid Q_2](\MM,\rho_x)=1$ iff $eval[Q_1](\MM,\rho_0,\rho_x)=1$ and $eval[Q_2](\MM,\rho_1,\rho_x)=1$, and (iii) for any $i\in\{0,1\}$, $eval[P\wedge Q_i](\MM,\rho_i,\rho_x)=1$ iff $eval[P](\MM)=1$ and $eval[Q_i](\MM,\rho_i,\rho_x)=1$.
\end{proof}

Similar to the equivalence relation ($\leftrightarrow$) in classical logic, its quantum analogue ($\Leftrightarrow$) is defined as follows: the expression $\phi\Leftrightarrow \psi$ is an abbreviation of the quantum formula $(y,|y|=1)[\phi\wedge \psi \mmid \neg^Q\phi\wedge \neg^Q\psi]$, where $y$ appears in neither $\phi$ nor $\psi$.

\begin{lemma}
$eval[\phi\Leftrightarrow\psi](\rho_y,\MM) =1$ iff the values of $eval[\phi](,\rho_y,\rho_y(0),\MM)$, $eval[\psi](\rho_y,\rho_y(0),\MM)$, $eval[\neg^Q\phi](\rho_y,\rho_y(1),\MM)$, and $eval[\neg^Q\psi](\rho_y,\rho_y(1),\MM)$ are all $1$, where $\rho_y(b)$ denotes $\alpha_b\qubit{b}$ for any bit $b\in\{0,1\}$ when $\rho_y=\sum_{u\in\{0,1\}}\alpha_u\qubit{u}$.
\end{lemma}

\begin{proof}
Assume that $eval[\phi\Leftrightarrow\psi](\rho_y,\MM)=1$. Since ``$\phi\Leftrightarrow \psi$'' is given by the use of the quantum OR, we obtain both $eval[\phi\wedge\psi](\rho_{y},\rho_y(0),\MM)=1$ and $eval[\neg^Q\phi\wedge \neg^Q\psi](\rho_{y},\rho_y(1),\MM)=1$. It then follows that the values of $eval[\phi](\rho_y,\rho_y(0),\MM)$, $eval[\psi](\rho_y,\rho_y(0),\MM)$,  $eval[\neg^Q\phi](,\rho_y,\rho_y(1),\MM)$, and $eval[\neg^Q\psi](\rho_y,\rho_y(1),\MM)$ are all $1$.
From these values, we conclude that $eval[\phi](\rho_y,\rho_y(0),\MM) =eval[\psi](\rho_y,\rho_y(0),\MM)$ and $eval[\neg^Q\phi](\rho_y,\rho_y(1),\MM) =eval[\neg^Q\psi](\rho_y,\rho_y(1),\MM)$.

On the contrary, we assume that the values of $eval[\phi](,\rho_y,\rho_y(0),\MM)$, $eval[\psi](\rho_y,\rho_y(0),\MM)$,  $eval[\neg^Q\phi](\rho_y,\rho_y(1),\MM)$, and $eval[\neg^Q\psi](\rho_y,\rho_y(1),\MM)$ are all $1$.
By the definition of the quantum AND, we obtain $eval[\phi\wedge \psi](\rho_y,\rho_y(0),\MM) =1$ and $eval[\neg^Q\phi\wedge \neg^Q\psi](\rho_y,\rho_y(1),\MM) =1$. From these values, we conclude that   $eval[\phi\Leftrightarrow \psi](\rho_y,\MM)=1$.
\end{proof}

In classical logic, the negation ($\neg$) can be suppressed to the elementary formulas. In a similar fashion, we can remove all quantum negations from each quantum formula.

\begin{definition}
A quantum formula $\phi$ is said to be \emph{negation free} if the quantum NOT ($\neg^Q$) does not appear in $\phi$.
\end{definition}

Definition \ref{evaluation} immediately yields the following simple fact.

\begin{lemma}\label{negation-free}
Given a quantum formula $\phi$, there exists another quantum formula $\phi'$ such that $\phi'$ is semantically equivalent to $\phi$ and $\phi'$ is negation free.
\end{lemma}

We intend to expand the quantum OR so that it can handle multiple qubits in the conditional part.

\begin{lemma}\label{multi-quantumOR}
Let $m\in\nat^{+}$ and let $Y_m = (y_1[e_1],y_2[e_2],\ldots,y_m[e_m])$. The expression $(Y_m)[P_{0^m}\mmid P_{0^{m-1}1}\mmid \cdots \mmid P_{1^m}]$ means that if $Y_m$ evaluates to the $k$th string $s_{k}^{(m)}$ of $\{0,1\}^m$, then $P_{s_{k}^{(m)}}$ is true. This expression is expressible within  $\mathrm{classicQFO}$.
\end{lemma}

\begin{proof}
We inductively define the desired quantum formula. When $m=1$, since $Y_1$ coincides with $y_1[e_1]$, $(Y_1)[P_{0}\mmid P_{1}]$ is the same as the (standard) quantum OR. Consider the case of $m\geq2$.
We define $(Y_m)[P_{0^m}\mmid P_{0^{m-1}1}\mmid \cdots \mmid P_{1^m}]$ to be $(y_1[e_1])[R_0 \mmid R_1]$ with $R_0\equiv (y_2[e_2],\ldots,y_m[e_m])[P_{0^m}\mmid P_{0^{m-1}1}\mmid \cdots \mmid P_{01^{m-1}}]$ and $R_1\equiv (y_2[e_2],\ldots,y_m[e_m])[P_{10^{m-1}}\mmid P_{10^{m-2}1}\mmid \cdots \mmid P_{1^{m}}]$. It is not difficult to show that this is equivalent to  $(Y_m)[P_{0^m}\mmid P_{0^{m-1}1}\mmid \cdots \mmid P_{1^m}]$.
\end{proof}

For later convenience, we call the above form of quantum OR for multiple qubits by the \emph{multi quantum OR}.

\begin{proposition}\label{multi-quantum-OR-prop}
The expression $(y[e_1,e_2])[\{P_u\}_{u\in\{0,1\}^m}]$ with $e_2-e_1+1=2^m$ means that, if $y[e_1,e_2]$ evaluates to a string $u\in\{0,1\}^m$, then $P_u$ is true. This expression is expressible within $\mathrm{classicQFO}$.
\end{proposition}

\begin{proof}
Note that $|y[e_1,e_2]|=2^m$. Hence, we set $Y_{2^m}=(y[e_1],y[e_1+1],\ldots,y[e_2])$ and consider $(Y_{2^m})[\{P_u\}_{u\in\{0,1\}^m}]$. Lemma \ref{multi-quantumOR} implies the desired proposition.
\end{proof}

Even in the framework of quantum first-order logic, it is possible to simulate the classical connectives $\wedge$ (AND) and $\vee$ (OR) indirectly.

\begin{lemma}
Assuming that $x,y,z,u,v,w$ are variables representing qubits.

(1) If these variables represent classical bits, then there exists a formula $R_{AND}$ such that $R_{AND}(x,y,z: u,v,w)$ is true iff $x=u$, $y=v$, and $w= (x\wedge y)\oplus z$ are all true.

(2) If these variables represent classical bits, then there exists a formula $R_{OR}$ such that $R_{OR}(x,y,z: u,v,w)$ is true iff $x=u$, $y=v$, and $w= (x\vee y)\oplus z$ are all true.
\end{lemma}

\begin{proof}
By Proposition \ref{multi-quantum-OR-prop}, we can freely use the multi quantum OR.

(1) We define $R_{AND}(x,y,z: u,v,w)$ to be $\PP_{I}(x: u)\wedge \PP_{I}(y: v)\wedge (\exists^Qp,|p|=1) ((x,y)[Q_{00}\mmid Q_{01}\mmid Q_{10}\mmid Q_{11}])$, where $Q_{s}\equiv \PP_{\oplus}(0,z:p,w)$ for all $s\in\{00,01,10\}$ and $Q_{11}\equiv \PP_{\oplus}(1,z:p,w)$.

(2) We define $R_{OR}(x,y,z: u,v,w)$ to be $\PP_{I}(x:u)\wedge \PP_{I}(y:v)\wedge (\exists^Q p,|p|=1) ( (x,y)[Q'_{00}\mmid Q'_{01}\mmid Q'_{10}\mmid Q'_{11}])$, where $Q'_{00}\equiv \PP_{\oplus}(0,z:p,w)$ and $Q'_{s}\equiv \PP_{\oplus}(1,z:p,w)$ for any $s\in\{01,10,11\}$.
\end{proof}

\vs{-1}
\subsection{Quantum First-Order Logic or QFO}\label{sec:define-QFO}

The \emph{quantum first-order language} $\LL$ is the set of quantum formulas built up from classical predicate symbols ($=$, $\leq$, $<$), quantum predicate symbols ($\PP_{I}$, $\PP_{ROT_{\theta}}$, $\simeq_{\varepsilon}$), quantum function symbols ($\otimes$, $QBIT$, $X$), and constant symbols ($0$, $1$, $\ilog(n)$, $n$) using quantum logical connectives ($\wedge$, $\|$, $\neg^Q$), variables ($i,j,k,\ldots, x,y,z,\ldots$),  classical quantifiers ($\forall$, $\exists$), and quantum quantifiers ($\forall^Q$, $\exists^Q$). For a  pure-state quantum variable $y$ and a classical/quantum variable $s$, we often use the abbreviation $y[s]$ for $QBIT(s,y)$.

Up to this point, we have finished preparing all necessary building blocks to formally introduce quantum first-order logic (QFO). We then wish to introduce the complexity class of promise (decision) problems expressible by quantum sentences with the instance function symbol $X$.
Remember that the instance function symbol $X$ is not treated as a predicate symbol as in the classical setting. Rather, $X$ is defined to be a unique function symbol, which corresponds to an ``input'' given to an underlying model of quantum Turing machines (quantum circuits or quantum recursion schemes).
A quantum formula with no free variable is called a \emph{quantum sentence}; in contrast, a quantum formula is said to be \emph{query-free} if it does not include the instance function symbol $X$.
A quantum sentence $\phi$ expresses a property of strings and each interpretation $\xi$ provides the meaning of classical/quantum terms and determines the (semantical) validity of the quantum sentence.
For a quantum formula $\phi$ with quantum terms $t$, $Mes(\phi)$ indicates the sum of all $\varepsilon$'s in the collection $m(\phi)$ of all subformulas of the form $t\simeq_{\varepsilon}0$ and $t\simeq_{\varepsilon}1$.
For example, if $m(\phi)$ includes $t_1\simeq_{\varepsilon_1}0$, $t_2\simeq_{\varepsilon_2}1$, and $t_3\simeq_{\varepsilon_3}0$, then $Mes(\phi)$ equals $\varepsilon_1+\varepsilon_2+\varepsilon_3$.

Recall the notion of iqq-freeness from Section \ref{sec:syntax}.

\begin{definition}[QFO]\label{definition-QFO}
A promise problem $(L^{(+)},L^{(-)})$ over the binary alphabet $\{0,1\}$ is said to be \emph{syntactically expressible} by a quantum sentence $\phi$ if there exists a fixed constant $\varepsilon_0\in[0,1/2)$ such that, for any given string $x\in\{0,1\}^*$,
(1) $Mes(\phi)\leq \varepsilon_0$ holds, and
(2) assuming that $x$ is assigned to the instance function symbol $X$ and  the length of $x$ is assigned to the constant symbol $n$, if $x\in L^{(+)}$, then these values make $\phi$ true, and if $x\in L^{(-)}$, then they make $\neg^Q \phi$ true.
We write $\mathrm{QFO}$ to denote the class of all promise problems  that are expressible by quantum sentences of quantum first-order logic. A subclass of $\mathrm{QFO}$, $\mathrm{classicQFO}$, is composed of  all promise  problems expressible by iqq-free quantum sentences.
\end{definition}

In Sections \ref{sec:character-QFO}--\ref{sec:functional}, we intend to study the characteristic natures of $\mathrm{QFO}$ and $\mathrm{classicQFO}$.

\vs{-1}
\section{Complexity Classes Characterized by QFO}\label{sec:character-QFO}

In Section \ref{sec:define-QFO}, we have introduced the quantum first-order logic, QFO, and its restriction, classicQFO.
We first argue the usefulness of quantum first-order logic in describing  certain complexity classes of promise problems.
In particular, we seek out a few applications of them to capture two quantum complexity classes $\mathrm{BQL}$ and $\mathrm{BQLOGTIME}$, which will be explained in the next  subsection.

Quantum Turing machines (QTMs) are capable of computing quantum functions mapping finite-dimensional Hilbert spaces to themselves. The use of a finite number of quantum quantifiers adds up a certain kind of alternation feature of existentialism and universality to those quantum functions.

Throughout this section, we explain our major contributions to the world of predicate logic.

\vs{-1}
\subsection{Logtime/Logspace Quantum Turing Machines}\label{sec:logtime-QTM}

Let us seek for a few useful applications of $\mathrm{QFO}$ and $\mathrm{classicQFO}$ to the characterization of certain types of quantum computing.
To explain how to express quantum computing, we first  explain two complexity classes $\bql$ and $\bqlogtime$ of \emph{promise problems} over the binary alphabet\footnote{In quantum complexity theory, underlying alphabets of QTMs are not necessarily limited to $\{0,1\}$. Here, the restriction onto $\{0,1\}$ is only meant to be in accordance with the definition of QFO and also to simplify the subsequent argument.} $\{0,1\}$.

Similar to a logtime deterministic Turing machine (DTM) \cite{BIS90},
a model of (poly)logtime quantum Turing machine (QTM) was considered in \cite{Yam22b} and was characterized in terms of a finite set of elementary quantum (recursion) schemes, which are variations of the recursion-theoretic schematic definition of \cite{Yam20}.
Notice that we use QTMs as a major computational model instead of quantum circuits because there is no clear-cut definition of quantum circuit to capture $\bqlogtime$ and the model of quantum circuits further requires a so-called ``uniformity'' notion.
To handle such a restricted model, nonetheless, we first recall a QTM of \cite{Yam22a}, which is equipped with multiple rewritable work tapes and \emph{random-access mechanism} that operates a tape head  along a rewritable index tape to access a read-only input tape.
More precisely, input qubits are initially given to the input tape and all the other tapes are initially set to be blank (with a special blank symbol, say, $B$) except for the left endmarkers placed at the ``start'' cells (i.e., cell $0$).
To access the $i$th input qubit, e.g., the machine must write $s^{(\ilog(n))}_i b$ ($b\in\{0,1\}$) on the index tape and enters a designated inner state called a \emph{query state}. In a single step, the input-tape head jumps to the $i$th tape cell and the last qubit $b$ is changed to $b\oplus x_{(i)}$ for a given input $x$.
It is important to remember that, after the input-tape head accesses the target tape cell, the index tape is not automatically erased and its tape head does not return to the start cell since erasing the index tape and moving the tape head back to cell $0$ cost a logarithmic number of steps.


Formally, a logtime QTM is defined as an octuple  $(Q,\Sigma,{\{\triangleright,\triangleleft\}}, \Gamma,\Theta, \delta, q_0,Q_{halt})$ with a finite set $Q$ of inner states, an input alphabet $\Sigma$, a work alphabet $\Gamma$, an index alphabet $\Theta=\{0,1\}$, a quantum transition function $\delta$, the initial (inner) state $q_0$, and a set $Q_{halt}$ of halting  (inner) states, where $\triangleright$ and $\triangleleft$ are two designated endmarkers.
A \emph{surface configuration} of $M$ on input $x$ is a sextuple $(q,l_1,l_2,w,l_3,z)$ in $Q\times \nat \times \nat \times \Gamma^{*} \times \nat \times \Theta^{*}$.
At each step, a QTM $M$ applies $\delta$, which maps  $Q\times\check{\Sigma}_{\lambda} \times \check{\Gamma}\times \check{\Theta} \times Q\times \check{\Gamma}\times \check{\Theta} \times D\times D$ to $\complex$, where $\check{\Sigma}_{\lambda}  =\Sigma\cup\{{\triangleright,\triangleleft},\lambda\}$, $\check{\Gamma}=\Gamma\cup\{{\triangleright}\}$,  $\check{\Theta}=\Theta\cup\{{\triangleright}\}$, and $D=\{-1,+1\}$.
For a surface configuration $(q,l_1,l_2,w,l_3,z)$, if $M$ makes a transition of the form $\delta(q,x_{(l_1)}, w_{(l_2)}, z_{(l_3)}, p, \tau,\xi,d_1,d_2) =\alpha$, then the next surface configuration becomes $(p,l_1,l_2+d_1,w', l_3+d_2,z')$ with (quantum) amplitude $\alpha$, where $w'$ and $z'$ are obtained from $w$ and $z$ by replacing $w_{(l_2)}$ and $z_{(l_3)}$ with $\tau$ and $\xi$, respectively.
We designate the first cell of the work tape as an output cell, which contains an ``output'' of the QTM.
The QTM begins with the initial inner state $q_0$ and eventually halts when it enters a halting state.
The quantum transition function $\delta$ of $M$ naturally induces the so-called \emph{time-evolution operator} $U_{\delta}$, which transforms  a surface configuration to another one obtained in a single step of $M$.
We demand $U_{\delta}$ to behave as a unitary transform in the space spanned by basis surface configurations of $M$ on an input.
This requirement is known as \emph{well-formedness} of the QTM. Hereafter, we implicitly assume that QTMs are always well-formed.
By setting $\Sigma=\{0,1\}$, we naturally expand the scope of inputs of QTMs to arbitrary quantum states.

\begin{figure}[t]
\centering
\includegraphics*[height=4.2cm]{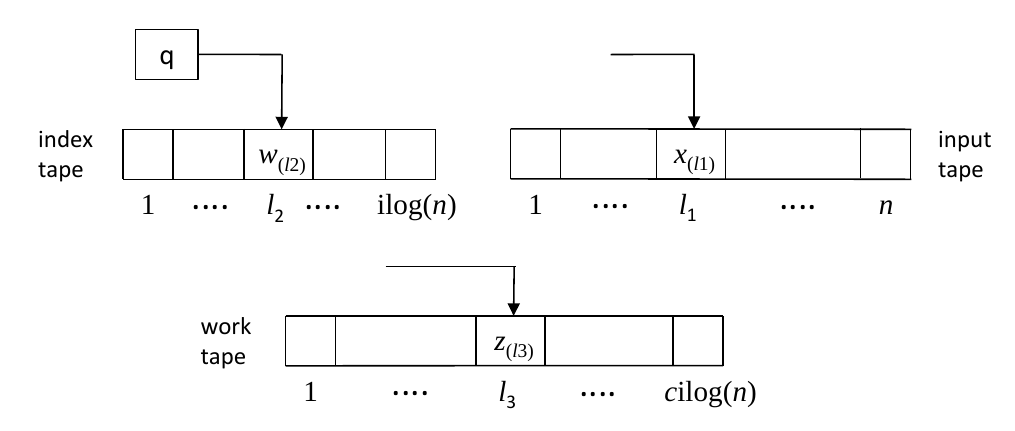}
\caption{A hardware of a QTM equipped with an input tape, an index tape, and a work tape.
}\label{fig:QTM-model}
\end{figure}

By conducting a measurement on the output cell, we can observe a classical bit, $0$ or $1$, with a certain probability. We say that, on input $\qubit{\phi}$, $M$ produces $b$ ($\in\{0,1\}$) with probability $\gamma$ if $b$ is observed with probability $\gamma$ by conducting a measurement on the output qubit.
A QTM $M$ is said to have  \emph{bounded errors} if there exists a constant $\varepsilon\in[0,1/2)$ such that, for any input $\qubit{\phi}$, when an output qubit, denoted $M(\qubit{\phi})$, is measured, we observe either $0$ with probability at least $1-\varepsilon$ or $1$ with probability at least $1-\varepsilon$.
For notational convenience, we write $\prob_{M}[M(\qubit{\phi})=b] = \gamma$ to express that $b$ is observed with probability $\gamma$ by $M$ on input $\qubit{\phi}$.

Given a promise problem $\LL=(L^{(+)},L^{(-)})$ over $\Phi_{\infty}$ and a real number $\eta\in[0,1]$, we say that a QTM $M$ \emph{solves $\LL$ with  probability at least $\eta$} if, for any input $\qubit{\phi}\in L^{(+)}$, $\prob_{M}[M(\qubit{\phi})=1]\geq \eta$ and, for any  $\qubit{\phi}\in L^{(-)}$, $\prob_{M}[M(\qubit{\phi})=0]\geq \eta$. When $\eta$ satisfies $\frac{1}{2}< \eta\leq 1$, $M$ is said to \emph{solve $\LL$ with bounded-error probability}.


Recall from Section \ref{sec:numbers} that numbers in $[0,n]_{\integer}$ are encoded into binary strings of length $\ilog(n)$ as:  $s^{(\ilog(n))}_0=0^{\ilog(n)}$, $s^{(\ilog(n))}_1 = 0^{\ilog(n)-1}1$, $s^{(\ilog(n))}_2 = 0^{\ilog(n)-2}10$, etc. We further encode tape symbols  $\{0,1,B\}$ of a QTM into 2-bit strings $\{\hat{0},\hat{1},\hat{B}\}$ by setting $\hat{0}=00$, $\hat{1}=01$, and $\hat{B}=11$.

We then define $\bqlogtime$ to be the collection of all promise problems solvable by QTMs
that run in $O(\log{n})$ steps (called \emph{logtime QTMs}, for short) with bounded-error probability (see \cite{Yam22b}).
We further expand $\bqlogtime$ to $\mathrm{HBQLOGTIME}$ by allowing a finite number of applications of ``limited'' quantum quantifiers
in the following recursive fashion.
A promise problem $(L^{(+)},L^{(-)})$ is in $\Sigma_k^Q\bqlogtime$ if there exist $k$ constants $c_1,c_2,\ldots,c_k$, another constant $\varepsilon\in[0,1/2)$, and a logtime QTM $M$ such that (i) for any string $x\in L^{(+)}$, $(\exists\qubit{\phi_1}\in \Phi_{2^{m_1}})(\forall \qubit{\phi_2}\in\Phi_{2^{m_2}})\cdots (Q_k\qubit{\phi_k}\in \Phi_{2^{m_k}}) \prob_{M}[M(\qubit{x},\Psi) =1]\geq 1-\varepsilon$ and (ii) for any string $x\in L^{(-)}$, $(\forall\qubit{\phi_1}\in \Phi_{2^{m_1}})(\exists \qubit{\phi_2}\in\Phi_{2^{m_2}})\cdots (\overline{Q}_k\qubit{\phi_k}\in \Phi_{2^{m_k}}) \prob_{M}[M(\qubit{x},\Psi) =0]\geq 1-\varepsilon$, where $\Psi = (\qubit{\phi_1},\qubit{\phi_2},\ldots, \qubit{\phi_k})$, $m_i= c_i\ilog(|x|)$ for all $i\in[k]$, $Q_k$ is $\exists$ (resp., $\forall$) if $k$ is odd (resp., even), and $\bar{Q}_k$ is $\exists$ (resp., $\forall$) if $Q_k=\forall$ (resp., $\exists$).
The desired class $\mathrm{HBQLOGTIME}$ is the union $\bigcup_{k\in\nat^{+}} \Sigma_k^Q\bqlogtime$, where the prefix ``H'' in $\mathrm{HBQLOGTIME}$ stands for ``hierarchy''.

In sharp contrast, we also consider \emph{logspace QTMs}, for which their runtime is not limited to log time.
The notation $\bql$ stands for the set of all promise problems solvable with bounded-error probability (say, $\geq 2/3$) by logspace QTMs with no runtime bound.

Unlike logtime QTMs, a logspace QTM cannot keep track of all moves on an $O(\log{n})$ space-bounded work tape. Consequently, such a QTM may not  postpone any measurement conducted in the middle of the computation until its final step.
In contrast, an extra space-unbounded ``garbage tape'' is used in \cite{Yam22a} to record these moves so that we require a measurement only at the end of each computation.

When logspace QTMs are required to run in \emph{expected polynomial time}, we use the notation ``$\ptime\bql$'' in place of $\bql$.
See, e.g.,  \cite{Yam22a} for more detail of this machine model.
In a similar fashion to build $\mathrm{HBQLOGTIME}$ from $\mathrm{BQLOGTIME}$, it is possible to expand $\mathrm{BQL}$ to $\mathrm{HBQL}$ by applying finite series of quantum quantifiers.

\begin{lemma}\label{class-relation}
$\mathrm{DLOGTIME} \subseteq \mathrm{BQLOGTIME} \subseteq \mathrm{PLOGTIME}_{\real}$, $\mathrm{PLOGTIME}_{\{0,1/2,1\}} \subseteq \mathrm{L}\subseteq \ptime\bql \subseteq \bql$, and $\mathrm{BQLOGTIME}\subseteq \ptime\bql$.
\end{lemma}

\begin{proof}
The inclusion $\ptime\bql\subseteq \bql$ is obvious by the definition.
Moreover, the inclusion
$\mathrm{BQLOGTIME}\subseteq \ptime\bql$ is obvious.
It is shown  in \cite{LMT00} that $O(\log{n})$-space deterministic Turing machine can be simulated by an appropriate $O(\log{n})$-space
reversible TM.
This concludes that  $\dl\subseteq \ptime\bql$.

To show that any logtime QTM $M$ can be simulated by an appropriate QTM $N$, we append an extra work tape so that, as $N$ simulates $M$, $N$ records a  surface configuration on this extra work tape.
This implies that $\mathrm{DLOGTIME} \subseteq \mathrm{BQLOGTIME}$ follows.

For the inclusion $\mathrm{PLOGTIME}_{\{0,1/2,1\}}\subseteq \dl$, we begin with taking any logtime (unbounded-error) probabilistic Turing machine (PTM) $M$. Since its computation paths on each input $x$ have length $O(\log{n})$, we can simulate each computation path of $M$ one by one using $O(\log{n})$ space-bounded work tape. To determine the acceptance or rejection of $M$, it suffices to count the total numbers of accepting computation paths and rejecting computation paths and then compare them. This entire procedure takes polynomial time using only $O(\log{n})$ space. Therefore, $L(M)$ is in $\dl$.

The remaining inclusion to prove is that $\mathrm{BQLOGTIME}\subseteq \mathrm{PLOGTIME}_{\real}$. This can be done by simple simulation techniques of \cite{ADH97,TYL10}. Here, we just sketch this simulation. Given a promise problem $\LL$ in $\mathrm{BQLOGTIME}$, we take a logtime QTM $M$ solving $\LL$ with bounded-error probability. The desired PTM $N$ works as follows. On input $x$, $N$ probabilistically pick two next moves $c_1$ and $c_2$ (including the case of $c_1=c_2$) with probability $|amp(c_1) amp(c_2)|$, where $amp(c)$ is $M$'s transition amplitude of selecting $c$.

If the simulated computation pairs are the same, then $N$ accepts the input $x$. Assume otherwise.
We calculate the product $\prod_{i=1}^{m}sgn(amp(c_i))$ of all signs of amplitudes of chosen moves $(c_1,c_2,\ldots,c_k)$ taken so far. We wish to simulate an arbitrary pair of computation paths. We remember such values of two computation paths. If two simulated paths are different, then we denote as $(s_1,s_2)$ the value pair obtained at the time of $M$'s halting. If $s_1s_2=+1$, then $N$ accepts $x$, otherwise, $N$ rejects it.

To make $N$ a proper probabilistic TM, we produce additional dummy moves with appropriate probabilities. For any computation path caused by such a move,  when $N$ terminates, we force it to accept and reject with equal probabilities.
\end{proof}

We further introduce an extension of $\bqlogtime$, denoted $\mathrm{BQLOG}^2\mathrm{TIME}$, by allowing $O(\log^2{n})$ runtime of underlying QTMs. Obviously, $\mathrm{BQLOGTIME}$ includes $\mathrm{BQLOG}^{2}\mathrm{TIME}$, but we do not know any relationships between $\mathrm{BQLOG}^2\mathrm{TIME}$ and $\mathrm{BQL}$.

\vs{-2}
\subsection{Complexity of QFO and classicQFO}

We discuss the descriptive complexity of $\mathrm{QFO}$ and $\mathrm{classicQFO}$ by finding the inclusion relationships of them to well-known time-bounded complexity classes.
Now, we begin the relationships between $\mathrm{classicQFO}$ and the logarithmic time-bounded complexity class.

\begin{proposition}\label{classicQFO-logtime}
$\mathrm{classicQFO} \subseteq \bqlogtime$.
\end{proposition}

\begin{proof}
Let us recall that quantum sentences used to define $\mathrm{classicQFO}$ are all iqq-free.  For any promise problem $\LL$ in $\mathrm{classicQFO}$ expressed by a certain iqq-free quantum sentence $\phi(X)$,
where $X$ is the instance function symbol,
its membership to $\bqlogtime$ can be verified by demonstrating how to simulate the evaluation process (explained in Section \ref{sec:semantics}) of quantum logical connectives and classical/quantum quantifiers on appropriate logtime QTMs.


Now, take any promise problem $\LL=(L^{(+)},L^{(-)})$ in $\mathrm{classicQFO}$ and let $\phi(X)$ denotes an iqq-free quantum sentence that semantically expresses $\LL$. To derive the desired membership of   $\LL$ to $\mathrm{BQLOGTIME}$, we examine the evaluation process of $\phi(X)$.

We first prepare one work tape for each pure-state quantum variable. This is possible because there are only a constant number of such variables in $\phi$. A key of our simulation strategy is that a quantum variable is translated into a series of qubits on a work tape of the QTM and that each classical term $i$ is translated into the tape head position on a work tape.

Our starting point is the case of quantifier-free quantum formulas. By Lemma \ref{negation-free}, we can assume that $\phi$ is negation free. Let us take a quantifier-free, negation-free quantum formula $\phi(X,i_1,i_2,\ldots,i_m, y_1,y_2,\ldots, y_k)$ with free variables $i_1,i_2,\ldots,i_m, y_1,y_2,\ldots,y_k$, where $A=(i_1,i_2,\ldots,i_m)$ is a series of classical variables and $Y=(y_1,y_2,\ldots,y_k)$ is a series of pure-state quantum variables.
We fix an interpretation of these variables as follows.
Let $\qubit{\xi}\in\Phi_{2^n}$, $\alpha =(\hat{i}_1,\hat{i}_2,\ldots,\hat{i}_m)$ with $\hat{i}_j\in[0,\ilog(n)]_{\integer}$, and $\Psi =(\qubit{\psi_1},\qubit{\psi_2}, \ldots,\qubit{\psi_k})$ with $\qubit{\psi_j}\in\Phi_{2^{\ilog(n)}}$, which respectively correspond to the symbols $X, i_1,\ldots,i_m, y_1,\ldots,y_k$.
Each element of $\Psi$ is expressed as a series of qubits on a work tape of $M$. Each element of $\alpha$ is expressed by the tape head position on a work tape.

\begin{claim}\label{quantifier-free}
There exists a constant-time QTM $M$ such that, for any $\qubit{\chi}$ and any $\Psi$,
(i) $\MM_{\chi,\Psi}\models_{\varepsilon} \phi(X,A,W)$ implies $\prob_{M}[M(\qubit{\chi}, \Psi)=1]\geq1-\varepsilon$ and (ii) $\MM_{\chi, \alpha, \Psi}\models_{\varepsilon} (\neg^Q)\phi(X,A,W)$ implies $\prob_{M}[M(\qubit{\chi}, \alpha, \Psi)=0] \geq 1-\varepsilon$, where $\MM_{\chi,\Psi}$ is a structure representing the behavior of $M$ on $\qubit{\chi}$ and $\Psi$.
\end{claim}

\begin{proof}
Associated with $\phi$, let us consider its variable connection graph $G=(V,E)$, which is topologically ordered.
We prove the claim by induction on the evaluation process (explained in Section \ref{sec:semantics}) of $\phi$, which
loosely corresponds to a constant-time QTM.

A classical term $suc(i)$ is simulated by the tape head, which moves to cell $\hat{i}+1$ from cell $\hat{i}$.
For a classical formula of the form $s=t$, we prepare two work tapes and generate $\hat{s}$ by moving a tape head to cell $\hat{s}$ and moving another tape head to cell $\hat{t}$ and then compare $\hat{s}$ and $\hat{t}$ by simultaneously moving two tape heads back toward cell $0$.

Since the quantum transform $ROT_{\theta}$ can be realized by a single move of an appropriate QTM, all quantum predicate symbols $\PP_F$ can be easily simulated by an appropriate constant-time QTM.

For quantum logical connectives, let us first consider the case of the quantum AND ($\wedge$) of the form $\phi\equiv R_1\wedge R_2$. Take two constant-time QTMs $M_1$ and $M_2$ for $R_1$ and $R_2$, respectively. We first combine two inputs given to $M_1$ and $M_2$. We then determine the order between $R_1$ and $R_2$ by considering the variable connection graph. Assume that $R_1$ comes before $R_2$. We run $M_1$ on this combine input and then run $M_2$ on the resulting quantum state. This is possible because $R_1$ and $R_2$ do not share the same variables in their first argument places.

Next, let us consider the case of the quantum OR ($\|$) of the form $\phi\equiv (z[s])[R_1\mmid R_2]$, which can be treated as the branching scheme of \cite{Yam20} for quantum functions. More precisely, assume by induction hypothesis that $R_1$ and $R_2$ are respectively simulatable on two constant-time QTMs $M_1$ and $M_2$. We define a new QTM $M$ to work as follows. When $s$ is a classical term interpreted to $\hat{s}$, we locate the $\hat{s}$-th qubit of a quantum state representing $z$. If $z[s]$ is $0$, then we run $M_1$; otherwise, we run $M_2$. On the contrary, when $s$ is a quantum term, the corresponding tape head is in a superposition. Hence, we apply the same action as explained above using each tape head position.

The quantum predicate symbol $\simeq_{\varepsilon}$ can be simulated by a (quantum) measurement performed at the end of computation.

Finally, let us consider an access to $X(z)$. Note that $|z|=\ilog(n)$. This can be conducted by an appropriate constant-time QTM as follows. Assume that the interpreted value of $z$ is generated on an index tape of the QTM. We then force the machine to enter a query state with $z$ so that its input-tape head jumps onto the target qubits of a given input assigned to $X$.
\end{proof}

Next, we focus on the case of classical quantifications.
When $\phi$ has the form $(\forall i\leq \ilog(n))R(i)$, we take a logtime  QTM $M$ simulating $R(i)$ for each natural number $\hat{i}$ assigned to the classical variable $i$.
The case of $\phi\equiv (\exists i\leq \ilog(n))R(i)$ is similarly handled.
The well-formedness ensures that, for instance, two quantifications of the form $(\forall i\leq \ilog(n))(\forall j\leq\ilog(n))R(i,j)$ or of the form $(\forall i\leq \ilog(n))(\exists j\leq\ilog(n))R(i,j)$, two classical variables $i$ and $j$ are ``independent'' in the sense that we inductively cycle through all numbers $\hat{i}$ up to $\ilog(n)$ by moving a tape head to the right and check the correctness of $R(i,j)$ by running $M$. This simulation can be done in $O(\log{n})$ steps.

For the case of a consequential existential quantifier over a quantum variable $y$, since $y$ is predecessor-dependent, the value of $y$ is uniquely determined from its predecessor by a certain quantum predicate symbol. This corresponds to the quantum state obtained after applying a unitary transform specified by the quantum predicate symbol.
\end{proof}


It is important to note that quantum quantifiers over quantum variables are interpreted as quantum quantifiers over qustrings of length  $\ilog(n)$.
Since $\mathrm{QFO}$ is obtained by applying the quantifications $\exists^Q$ and $\forall^Q$ to underlying quantum formulas used in $\mathrm{classicQFO}$, we can draw the following conclusion from Proposition  \ref{classicQFO-logtime}.

\begin{theorem}\label{QFO-HQLOGTIME}
$\mathrm{QFO}\subseteq \mathrm{HBQLOGTIME}$.
\end{theorem}

\begin{proof}
Let $\LL=(L^{(+)},L^{(-)})$ denote any promise problem in $\mathrm{QFO}$ and let $\phi$ be a quantum sentence that syntactically expresses $\LL$.

We partition all variables into two groups. The first group consists of predecessor-dependent quantum variables and the second group consists of predecessor-independent quantum variables.

It is possible to reformulate $\phi$ so that it has the form $(\exists^Qy_1,|y_1|=s_1)(\forall^Qy_2,|y_2|=s_2) \cdots (Q_ky_k,|y_k|=s_k) \psi(X,y_1,y_2,\ldots,y_k)$, where $\psi$ is obtained from a quantifier-free quantum formula whose predecessor-dependent quantum variables are quantified by $\exists^Q$, $y_1,y_2,\ldots,y_k$ are predecessor independent,  $s_1,s_2,\ldots,s_k$ are classical terms, and $Q_k=\forall^Q$ if $k$ is even and $\exists^Q$ otherwise. By Proposition \ref{classicQFO-logtime}, $\psi$ is simulated on a logtime QTM. Therefore, $\phi$ is simulatable by additional quantifications, and thus it belongs to $\mathrm{HBQLOGTIME}$.
\end{proof}


Next, we wish to compare between the expressibilities of $\mathrm{FO}$ and $\mathrm{QFO}$. As one expects, quantum first-order logic is powerful enough to express all languages in $\mathrm{FO}$.

\begin{theorem}\label{FO-vs-QFO}
$\mathrm{FO}\subseteq \mathrm{QFO}$.
\end{theorem}

\begin{proof}
Assume that classical predicates take natural numbers. To interpret natural numbers to quantum strings, since quantum predicates work with superpositions of binary strings, we intend to identify natural numbers with binary strings.

Let us consider any classical first-order sentence $\phi$ of the form $(\exists x_1)(\forall x_2)\cdots (Q_k x_k) P(x_1,x_2,\ldots,x_k)$, where $P$ is a quantifier-free formula containing (possibly) a unique atomic predicate $X(\cdot)$.
By appropriate applications of DeMorgan's law to $\phi$, we assume that classical negation ($\neg$) appears only in the form of $\neg X(i)$.
Take any classical binary input string $w$, which is viewed as a structure for $X(\cdot)$.

We first get rid of the elementary formulas of the form $s=t$ and $s\leq t$ as follows. Instead of using $s=t$, we use the formula $(\forall i)[BIT(i,s)\leftrightarrow BIT(i,t)]$. For $s\leq t$, we use $(\exists i)[(BIT^*(i,t)\wedge \neg BIT^*(i,s))\wedge (\forall j<i)(BIT^*(j,s)\leftrightarrow BIT^*(j,t))]$, where $BIT^*(j,s)$ means that the $j$th bit of the binary representation of $s$ extended by adding extra preceding $0$s to make the length be exactly $\ilog(n)$. For example, when $n=16$, the extended binary representation of $s=3$ is $0011$.

Let $S_w$ denote a structure representing input string $w$.

(1) We focus on the first case where $\phi$ is quantifier-free. We aim at building a new quantum formula $\phi^Q$ satisfying the equivalence:
$\SSS_w\models \phi$ iff $\SSS_w\models \phi^Q$.
Since $\phi$ is constructed from atomic formulas by applying logical connectives and quantifiers, we show the theorem by induction on the construction process of $\phi$.

For our purpose, we first translate $\phi$ into $\phi^Q$ in the following inductive  way. Throughout the subsequent argument, we set $\varepsilon=0$.
Remember that, in the classical setting, $X$ and $BIT$ are ``predicate symbols'' whereas, in the quantum setting, $X$ and $QBIT$ are ``function'' symbols''.
The classical predicate $X(i)$ is simulated by $X(i)\simeq_{0} 1$ in our quantum logic and $\neg X(i)$ is simulated by $X(i)\simeq_{0} 0$.
Similarly, the classical predicates $BIT(x,y)$ and $\neg BIT(x,y)$ are simulated by $QBIT(x,z)\simeq_{0} 1$ and $QBIT(x,z)\simeq_{0} 0$, provided that number $y$ is expressed as binary string $(z[1],z[2],\ldots,z[\ilog(n)])$ since $y\in[n]$ and that $z$ is a quantum variable of $\ilog(n)$ qubits associated with $y$.

Consider the case where $\phi$ is of the form $P_1\wedge P_2$.
By induction hypothesis, it follows that, for any $a\in\{1,2\}$, $S_w\models P_a$ iff $S_w\models P_a^Q$.
In the case where $P_1$ and $P_2$ share no common free variables,
we simulate $P_1\wedge P_2$ by $\phi^Q\equiv P_1^Q \wedge P_2^Q$.
We thus conclude that  $S_w\models \phi$ iff $S_w\models \phi^Q$.
We then consider the other case where $P_1$ and $P_2$ share common free variables,
say, $x_1,x_2,\ldots,x_k$. In this case, we write $P_1(x_1,\ldots,x_k)$ and $P_2(x_1,\ldots,x_k)$ by emphasizing the shared free variables (and by  suppressing all the other non-shared variables). We define $\phi^Q$ to be $(\exists^Q z_1)\cdots (\exists^Q z_k)[P_1^Q(x_1,\ldots,x_k)\wedge (\bigwedge_{i=1}^{k} \PP_{I}(x_i:z_i))\wedge P_2^Q(z_1,\ldots,z_k)]$.

Next, we assume that $\phi$ has the form $P_1\vee P_2$. We simulate $P_1\vee P_2$ by $\phi^Q\equiv (\exists^Q y,|y|=1)[(y)[P_1^Q\mmid P_2^Q]]$, where $y$ is a ``new'' bound pure-state quantum variable not used in any other place.
Note that $\phi^Q$ is iqq-free. By induction hypothesis, it follows that for any $a\in\{1,2\}$, $S_w\models P_a$ iff $S_w\models P^Q_a$. Since $S_w\models P_1\vee P_2$ iff $S_w\models P_1$ or $S_w\models P_2$, we obtain the desired result.

(2) We consider quantifiers that are applied to a quantifier-free formula $P$. By (1),  it follows that $S_w\models P$ iff $S_w\models P^Q$.
Assume that $\phi$ is of the form $(\exists i_1\leq\ell_1(n))(\forall i_2\leq \ell_2(n))\cdots (Q_ki_k\leq\ell_k(n)) P(X,i_1,i_2,\ldots,i_k)$. Note that $\ell_i(n)\leq\ilog(n)$.
We use quantum variables $y_1,\ldots,y_k$, each $y_i$ of which  represents  elements in  $\Phi_{2^{\ell_i(n)}}$ to simulate $i_1,\ldots,i_k$. We then set $\phi^Q$ to be $(\exists^Q y_1,|y_1|=\ell_1(n))(\forall^Q y_2,|y_2|=\ell_2(n))\cdots (Q^Q_k y_k,|y_k|=\ell_k(n)) P^Q(X,y_1,y_2,\ldots,y_k)$.

If $X$ is evaluated by input $w$, then we obtain either $S_w\models \phi$ or $S_w\not\models \phi$. By the above definitions, we conclude that either $S_w\models \phi^Q$ or $S_w\not\models\phi^Q$.
\end{proof}

Since $\mathrm{FO}$ equals $\mathrm{HLOGTIME}$ \cite{BIS90}, Theorem \ref{FO-vs-QFO} leads to another conclusion of $\mathrm{HLOGTIME}\subseteq \mathrm{QFO}$.


\vs{-1}
\subsection{Quantum Transitive Closure or QTC}

Since $\mathrm{FO}$ is quite weak in expressing power, several supplementary operators have been sought in the past literature to enhance the expressibility of $\mathrm{FO}$.
For instance, the least fixed point operator and the total ordering on universe were used to expand the expressibility of logical systems (see, e.g., \cite{Imm86}).
One of such effective operators is a \emph{transitive closure (TC) operator}, which  intuitively works as a restricted type of recursion scheme.
While $PLUS$ (addition) is expressible by $BIT$ and $=$ (equality) \cite{BIS90,Imm87a}, $BIT$ is expressible by $PLUS$ and $=$ in the presence of $TC$-operators.
See, e.g., \cite{Imm86} for its usage.

Concerning $\mathrm{QFO}$, we also wish to seek out a ``feasible'' quantum analogue of such a transitive closure operator. We conveniently call it the \emph{quantum transitive closure (QTC) operator}.
Since we cannot allow us to apply the QTC operator to arbitrary quantum relations, we need to limit the scope of QTC to quantum bare relations.
A \emph{quantum bare relation} $R$ of arity $k+m+2$ is a measurement-free, iqq-free quantum formula of the form $R(i,x_1,x_2,\ldots,x_k: j,y_1,y_2,\ldots,y_m)$ with $k+m+2$ free variables, where $i,j$ are classical variables, $x_1,x_2,\ldots,x_k$ are predecessor-independent quantum variables, and $y_1,y_2,\ldots,y_m$ are predecessor-dependent quantum variables. We further require $i>j$ for the termination of the transitivity condition of $QTC$.
A simple example is the formula of the form $i=j+1$.
Let $\phi$ denote any quantum bare relation. We define $QTC[\phi]$ to be the reflexive, transitive closure of $\phi$ (if any).
As a quick example, define $\phi(i,y:j,z) \equiv \PP_{ROT_{\pi/4}}(y[i]:z[i]) \wedge i=j+1 \wedge i\leq \ilog(n)$ and consider $\phi'\equiv QTC[\phi](\ilog(n),y:0,z)$. This provides the quantum effect similar to $\PP_{WH}$ on $y$ and $z$ in Example \ref{example-formula}.

\begin{definition}
The notation $\mathrm{classicQFO} +QTC$ (resp., $\mathrm{QFO}+QTC$) refers to the set of promise problems  expressible by $\mathrm{classicQFO}$ (resp., $\mathrm{QFO}$) together with the free use of the $QTC$ operator for quantum bare relations.
We also consider a restriction of $QTC$ specified by the following condition: $QTC$ satisfies that the start value of $i$ is upper-bounded by $c\ilog(n)$ for an absolute constant $c\in\nat^{+}$.
In this case, we emphatically write $logQTC$.
\end{definition}

Let us recall the complexity class $\mathrm{BQLOG}^2\mathrm{TIME}$ from Section \ref{sec:logtime-QTM}.
For a quantum variable $u$ and two numbers $i,j\in\nat^{+}$ with $i\leq j$, hereafter, we write $u[i,j]$ in place of $u[i]\otimes u[i+1] \otimes u[i+2] \otimes \cdots \otimes u[j]$. Hence, the qubit size of $u[i,j]$ is $\hat{j}-\hat{i}+1$ for the numbers $\hat{i}$ and $\hat{j}$ assigned to $i$ and $j$, respectively.


\begin{proposition}\label{QFO+logQTC-bounds}
$\bqlogtime \subseteq \mathrm{classicQFO}+logQTC \subseteq  \mathrm{BQLOG}^2\mathrm{TIME}$.
\end{proposition}

\begin{proof}
We split the statement of the proposition into two claims: (1) $\mathrm{classicQFO}+logQTC \subseteq  \mathrm{BQLOG}^2\mathrm{TIME}$ and (2) $\bqlogtime \subseteq \mathrm{classicQFO}+logQTC$. In what follows, we prove these claims separately.


(1) Since $\mathrm{classicQFO}\subseteq \bqlogtime$ by Proposition  \ref{classicQFO-logtime}, it suffices to simulate the behavior of  the $logQTC$-operator  on a certain logtime QTM. Let us consider a quantum formula of the form $QTC[\phi](i,s:j,t)$ for a quantum bare relation $\phi$. Since $\phi$ is measurement-free and iqq-free, by Proposition  \ref{classicQFO-logtime},
$\phi(i,s:j,t)$ can be simulatable by a certain logtime QTM, say, $N$.
It is important to note that, during the simulation of the  $logQTC$-operator, we do not need to make
any query to an input given to $N$.

To simulate $QTC[\phi](i,s:j,t)$ on a QTM, let us consider the following procedure. Initially, we set $k$ to be $i$ and then, by decrementing $k$ down to $j$ by one, we recursively simulate  $\phi$. Notice that the simulation can be conducted at most $i$ times. Since $i\leq O(\log{n})$, this entire procedure takes $O(\log^2{n})$ steps.

(2) Let us consider any promise problem $\LL=(L^{(+)},L^{(-)})$  over the binary alphabet in $\bqlogtime$ and take a logtime QTM $M$ that solves this problem $\LL$ with high probability.
It is possible to increase the size of $Q$ so that  $M$ takes exactly two next moves at any step.
Recall that a surface configuration of $M$ is of the form $(q,l_1,l_2,w,l_3,z)$, which expresses that $M$ is in inner state $q$, scanning the $l_1$th cell of the input tape, the $l_2$th cell of the index tape containing string $w$, and the $l_3$th cell of the work tape  composed of string $z$, where $q\in Q$, $l_1\in[0,n+1]_{\integer}$, $l_2\in[0,\ilog(n)+1]_{\integer}$, $l_3\in[0,c\ilog(n)]_{\integer}$, and $w,z\in\{0,1,B\}^{*}$.

Assume further that the runtime of $M$ is at most $c\ilog(n)$  for a certain fixed constant $c\in\nat^{+}$, and thus $M$'s tape space usage is at most $c\ilog(n)$.
We first modify $M$  by installing an \emph{internal clock} so that $M$'s computation paths halt at the same time (except for all computation paths of zero amplitude). At the end of computation, we observe the first qubit of its output to determine that $M$ accepts or rejects a given input.

For convenience, we introduce a \emph{pinpoint configuration} of the form $(q,x_{(l_1)},w_{(l_2)},z_{(l_3)},l_1,l_2,l_3)$, which is naturally induced from a surface configuration $(q,l_1,l_2,w,l_3,z)$ and a given input $x$.
Similarly to the time-evolution operator, the quantum transition function $\delta$ of $M$ induces the quantum function $STEP$, which is defined as
$STEP(\qubit{q,x_{(l_1)},w_{(l_2)},z_{(l_3)},l_1,l_2,l_3}) = \cos\theta  \qubit{p,x_{(l_1)},\tau,\xi,l_1,l_2+d_1,l_3+d_2} + \sin\theta \qubit{p',x_{(l_1)},\tau,\xi,l_1,l_2+d'_1,l_3+d'_2}$ for a fixed constant $\theta\in[0,2\pi)$.

For simplicity, we assume that $Q=[2^{e}]$ for a certain fixed constant $e\in\nat^{+}$.
We encode $M$'s tape symbols into binary strings as $\hat{0}=00$, $\hat{1}=01$, and $\hat{B}=11$.
An entire  tape content is then expressed by a series of such encoded symbols. For example, $000111$ indicates the tape content of $01B$.
We express each element of $Q$ as a binary string of length $e$. In particular, $q_0=0^e$ follows.
Assume that $l_1$ is expressed in binary as $s^{(\ilog(n))}_{l_1}$, and $l_2$ and $l_3$ are expressed in unary as $0^{l_2-1}10^{\ilog(n)-l_2}$ and  $0^{l_3-1}10^{c\ilog(n)-l_3}$, respectively.

We introduce pure-state quantum variables $(s,r_1,r_2,r_3,y_1,y_2)$ evaluating to $(q,l_1,l_2,l_3,w,z)$. Moreover, $ans$ denotes a 1-qubit quantum variable expressing a query answer from $X$.
It is remarked here that $M$'s computation is described as a series of surface configurations, in which any inner state $q$, a query answer $x_{l_1)}$, the tape contents $w$ and $z$, and their tape head positions $l_2$ and $l_3$ are quantumly ``entangled'' in general, and consequently, they are ``inseparable''. For instance, we may not use $l_2$ and $l_3$ separately from $w$ and $z$. Hence, we need to treat $(s,ans,y_1,y_2,r_1,r_2,r_3)$ as a single quantum state.
We thus set $Z$ and $Z'$ to denote two quantum terms $s\otimes ans \otimes r_1\otimes r_2 \otimes y_1\otimes r_3\otimes y_2$  and $s'\otimes ans \otimes r'_1\otimes r'_2\otimes y'_1\otimes r'_3\otimes y'_2$, respectively.
The last qubit of $y_1$ is used for an answer from $X$ to a query word $b$.
Since $|s|=e$, $|r_1|=|r_2|=\ilog(n)$, $|y_1|=2\ilog(n)$,   $|r_3|=c\ilog(n)$, and $|y_2|=2c\ilog(n)$, it follows that $|Z|=|Z'|=e+(4+3c)\ilog(n)+1$.
We then $Z$ to specify quantum terms $(s,ans,r_1,r_2,y_1,r_3,y_2)$ whenever we need them. For instance, $y_1$ is specified by  $Z[e+2\ilog(n)+2,e+4\ilog(n)+1]$.

We introduce a quantum variable $V$ with $|V|=e+6$, which evaluates to $(p,\tau,\xi,\hat{d}_1,\hat{d}_2)$ with $p\in Q$, $\tau,\xi\in\{0,1\}^*$, and $\hat{d}_1,\hat{d}_2\in\{0,1\}$, where $\hat{d}_1$ and $\hat{d}_2$ are defined to satisfy that
$d_1=(-1)^{\hat{d}_1}$ and $d_2=(-1)^{\hat{d}_2}$.

A basic idea of the following argument is that we syntactically express a single transition of $M$ by $COMP$ and apply $QTC$ to repeat $COMP$  $c\ilog(n)$ times to simulate the entire computation of $M$.

We start with defining the quantum formula $INI(Z)$ to mean that $(s,ans,y_1,y_2,r_1,r_2,r_3)$ takes the initial value $(q_0,0,\hat{B}^{\ilog(n)+1},\hat{B}^{c\ilog(n)+1},0,0,0)$.
We then introduce $COMP(Z:Z')$ to describe $STEP$ in the following way: when   $(s,X(r_1),y_1[r_2],y_2[r_3],r_1,r_2,r_3)$ takes a value $(q,x_{(l_1)},w_{(l_2)},z_{(l_3)},l_1,l_2,l_3)$, $STEP$ changes it to $(p,x_{(l_1)},\tau,\xi,l_1,l_2+d_1,l_3+d_3)$ and $(p',x_{(l_1)},\tau,\xi,l_1,l_2+d'_1,l_3+d'_3)$ with amplitudes $\cos\theta$ and $\sin\theta$, respectively.

The remaining task is therefore to construct $COMP$ within our logical system $\mathrm{classicQFO}$.

\begin{claim}
$COMP$ is expressible in $\mathrm{classicQFO}$.
\end{claim}

\begin{proof}
Firstly, we concentrate on the behavior of $M$'s index tape and its tape head. Assume that $M$'s pinpoint configuration at time $t$ is of the form  $(q,x_{(l_1)},w_{(l_2)},z_{(l_3)},l_1,l_2,l_3)$. After the application of $STEP$, it changes to $(q,x_{(l_1)},\tau,\xi,l_1,l_2+d_1,l_3+d_2)$.
We intend to express this step logically.

We set $sym_i$ and $hp_i$ to denote the $i$th cell content of the index tape and the presence of the tape head (i.e., $1$ for ``YES'' and $0$ for ``NO''), namely, $Z[e+2\ilog(n)+2i,e+2\ilog(n)+2i+1]$ and $Z[e+\ilog(n)+i+1]$, respectively.
Moreover, we set $newsym_i$ and $newhp_i$ to respectively denote the tape symbol and the presence of the tape head at the next step, that is, $Z'[e+2\ilog(n)+2i,e+2\ilog(n)+2i+1]$
and $Z'[e+\ilog(n)+i+1]$, respectively.
Since the tape head moves affect only three consecutive tape cells at each step, it suffices to consider only three consecutive cells $i-1$, $i$, and $i+1$.

To determine the value of $STEP$, we first obtain the values of $(p,\tau,\xi,d_1,d_2)$ from $(q,x_{(l_1)},w_{(l_2)},z_{(l_3)})$ by a single application of $\delta$. For this reason, we set $C$ to express the quantum terms $(s,X(r_1),y_1[2r_2-1,2r_2],y_2[2r_3-1,2r_3])$, which evaluates to $(q,x_{(l_1)},w_{(l_2)},z_{(l_3)})$ and also set $V$ to evaluate to $\hat{d}_1\hat{d}_2$.

In the subsequent description, we wish to differentiate between query steps and non-query steps. In a non-query step, we intend to handle the index tape and the work tape separately and then define the corresponding quantum formulas $INDEX$ and $WORK$. In a query step, on the contrary, we intend to handle an access to the given input $x$ and then define the associated quantum formula $QUERY$.

\ms

(I)
We first consider the case where the current step is a non-query one. Meanwhile, we focus on the index tape.
Let us consider the case where $C$ takes a value $(q,x_{(l_1)},w_{(l_2)},z_{(l_3)})$. Notice that $V$ is obtained by applying $\PP_{ROT_{\theta}}(00:V)$, which means that the quantum transform $ROT_{\theta}$ changes $\qubit{00}$ to $\cos\theta\qubit{\hat{d}_1\hat{d}_2} + \sin\theta\qubit{\hat{d}'_1\hat{d}'_2}$ according to $STEP$.
In what follows, we assume that $V$ evaluates to $\hat{d}_1\hat{d}_2$ with amplitude $\cos\theta$. The case of $\hat{d}'_1\hat{d}'_2$ is similar.

Assume that three consecutive values $(hp_{i-1},hp_i,hp_{i+1})$ respectively evaluate to the presence of the tape head on cells $i-1$, $i$, and $i+1$, each of which takes one of the following values: $(0,0,0)$, $(0,0,1)$, $(0,1,0)$, and $(1,0,0)$.

When $C$ evaluates to $(q,x_{(l_1)},w_{(l_2)},z_{(l_3)})$,

(1) When $(hp_{i-1},hp_i,hp_{i+1})$ takes the value $(0,0,1)$, no matter how $M$ behaves, the content of cell $i$ never changes. However, its tape head presence changes whenever the tape head moves to the left. More precisely, if $\hat{d}_1=1$, then we set $Q_{(1)}$ to be $\PP_{copy}(sym_{i}: newsym_{i}) \wedge \PP_{copy}(1:newhp_i)$; however, when $\hat{d}_1=1$, $Q_{(1)}$ is set to be $\PP_{copy}(sym_i:newsym_i)\wedge \PP_{copy}(0:newhp_i)$.
We then leave all the other variables undefined (because they are determined by other cases).

(2) Similarly, when $(hp_{i-1},hp_i,hp_{i+1})$ takes the value $(1,0,0)$, the content of cell $i$ never changes.
When $\hat{d}_1=1$ (resp., $\hat{d}_1=0$),  $Q_{(2)}$ is set to be $\PP_{copy}(sym_{i}:newsym_{i}) \wedge \PP_{copy}(0:newhp_i)$ (resp.,  $\PP_{copy}(sym_i:newsym_i)\wedge \PP_{copy}(1:newhp_i)$) and all the other variables are left undefined.

(3) In the case where $(hp_{i-1},hp_i,hp_{i+1})$ takes the value $(0,1,0)$, since the tape head moves away from cell $i$, we set $Q_{(3)}$ to be $\PP_{copy}(\hat{\tau}:newsym_i) \wedge \PP_{copy}(0:newhp_i)$.

(4) In sharp contrast, when $(hp_{i-1},hp_i,hp_{i+1})$ takes the value $(0,0,0)$, the content of cell $i$ and the presence of the tape head on cell $i$ do not change. It thus suffices to set $Q_{(4)}$ to be $\PP_{copy}(sym_{i},hp_{i}:newsym_{i},newhp_{i})$ and leave others undefined.

(5) Finally, we set $INDEX$ to denote $(\forall i\leq\ilog(n))((hp_{i-1},hp_i,hp_{i+1},\hat{d}_1)[ \{Q_{(j)}\}_{j}\cup \{Q'\}])$, where $Q'\equiv false$ corresponds to all the other cases.

In a similar fashion, we can treat the work tape of space bound $c\ilog(n)$ and then define $WORK$.

\ms

(II)
Next, we consider the case of a query step. We wish to logically express an access to the input $x$. When $M$'s inner state $q$ is $q_{query}$, $M$ makes a query to the input.
This can be done logically using $X(y_1)$ with the use of the quantum OR. Formally, we set $QUERY\equiv (X(y_1))[\PP_{copy}(0:newns)\mmid \PP_{copy}(1:newans)]$, where $ans$ and $newans$ are respectively related to $Z$ and $Z'$.
\end{proof}

In the end, $COMP$ is given as the quantum formula indicating that, if $s$ evaluates to $q_{query}$, then $QUERY$ is true, and otherwise, $INDEX\wedge WORK$ is true. This quantum formula $COMP$ is obtainable by the use of the multi quantum OR.
To simulate the entire computation of $M$ on input $x$, we need to repeat $COMP$ $c\ilog(n)$ times with the condition that, whenever $s$ evaluates to a halting state $q_{halt}$, we force to ``copy'' $Z$ to $Z'$.
This can be done by the use of the QTC operator as $QTC[COMP'](t,Z:t',Z')$, where $COMP'(t,Z:t',Z')\equiv COMP(Z,Z')\wedge t=t'+1\wedge t\leq c\ilog(n)$.

When $M$ halts, we apply a measurement on the first qubit of the work tape to know whether $M$ accepts or rejects the input. Logically, we use the quantum formula of the form $y'_2[2]\simeq_{\varepsilon}1$, where $y'_2$ comes from $Z'$.

This completes the proof of the proposition.
\end{proof}


As an immediate consequence of Proposition \ref{QFO+logQTC-bounds}, we obtain the following upper and lower bounds on the computational complexity of  $\mathrm{QFO}+logQTC$.

\begin{corollary}\label{QFO-case-logtime}
$\mathrm{HBQLOGTIME} \subseteq \mathrm{QFO}+ log{QTC} \subseteq \mathrm{HBQLOG}^2\mathrm{TIME}$.
\end{corollary}

\begin{proof}
By Proposition \ref{QFO+logQTC-bounds}, we obtain $\mathrm{BQLOGTIME} \subseteq \mathrm{classicQFO}+logQTC \subseteq \mathrm{BQLOG}^2\mathrm{TIME}$. Since $\mathrm{QFO}$ is obtained from $\mathrm{classicQFO}$ by applying quantum quantifiers, this fact yields the corollary following an argument similar to the proof of Theorem \ref{QFO-HQLOGTIME}.
\end{proof}


Next, we examine the full power of QTC operator.

\begin{theorem}\label{QFO+QTC-vs-BQL}
$\mathrm{classicQFO} +QTC = \ptime\bql$.
\end{theorem}

\begin{proof}
Hereafter, we intend to prove the following two claims separately: (1) $\mathrm{classicQFO} +QTC \subseteq \ptime\bql$ and (2) $\ptime\bql \subseteq \mathrm{classicQFO} +QTC$.

(1)
Since $\mathrm{classicQFO}\subseteq \bqlogtime$ by Proposition  \ref{classicQFO-logtime}, Lemma \ref{class-relation} implies that  $\mathrm{classicQFO}\subseteq \ptime\bql$. It thus suffices to simulate $QTC$ on an appropriately chosen  logspace QTM running in polynomial time. Let us consider a quantum formula of the form $QTC[R](i,s:j,t)$, provided that $R$ can be simulatable by a certain logspace QTM, say, $N$. We implement a ``counter'' on the desired QTM using an extra work tape.
The simulation of $QTC[R](i,s:j,t)$ starts with $k=i$ in this counter and repeatedly runs $N$ by incrementing $k$ up to $j$ in the counter. This repetition takes $O(n)$ times. Note that the counter can be realized by an  $O(\log{n})$ space-bounded work tape. Thus, the whole simulation can be done in polynomial time using only $O(\log{n})$ space. This implies the desired inclusion.

(2)
Let $\LL$ denote any promise problem in $\ptime\bql$ over the binary alphabet. We take a logspace QTM $M$ that solves $\LL$ with bounded-error  probability in expected polynomial time.
It is possible to choose an appropriate constant polynomial $p(n)$ so that, even if we force $M$ to cut off its computation after $p(n)$ steps pass and  then treat all unterminated computation paths ``undefined'', either the acceptance probability or the rejection probability still exceeds $1-\varepsilon$ for a certain fixed constant $\varepsilon\in[0,1/2)$. Tentatively, we call this $p(n)$ a \emph{critical time bound} of $M$.

In what follows, we fix a critical time bound of $M$ to be $n^c$ for an appropriate constant $c\geq1$ and we further assume that the work space usage of $M$ is at most $c\log n$.
We first modify $M$ so that, by installing an \emph{internal clock}, $M$'s computation paths halt at the same time (except for all computation paths of zero amplitude or of exceeding the critical time bound).

A basic idea behind the simulation of $M$ is similar to the proof of Proposition \ref{QFO+logQTC-bounds}. Here, we only mention the points that are different from it.
Following the proof of the proposition, we can define the quantum formula $COMP$ to describe a single move of $M$.
Since $n^c$ is the critical time bound of $M$, it suffices to repeat $COMP$ $n^c$ times to simulate an entire computation of $M$.
This can be done by the use of the QTC-operator.
More precisely, we set $QTC[COMP'](t,Z:t',Z')$, where $COMP'\equiv COMP(Z,Z')\wedge t=t'+1\wedge t\leq n^c$.

This completes the proof of the theorem.
\end{proof}


Theorem \ref{QFO+QTC-vs-BQL} can be naturally extended into $\mathrm{QFO}$ and $\ptime\mathrm{HBQL}$. We then obtain the following corollary.

\begin{corollary}\label{HBQL-QFO+QTC}
$\mathrm{QFO}+QTC = \ptime\mathrm{HBQL}$.
\end{corollary}

\begin{proof}
Let us first claim that $\mathrm{QFO}+ QTC \subseteq \ptime\mathrm{HBQL}$.
Theorem \ref{QFO-HQLOGTIME} yields $\mathrm{QFO}\subseteq \mathrm{HBQLOGTIME}$. From this inclusion, we obtain $\mathrm{QFO}\subseteq \ptime\mathrm{HBQL}$.
As in the proof of Theorem \ref{QFO+QTC-vs-BQL}, $QTC$ can be simulated by logspace QTMs in polynomial time. Thus, we conclude that $\mathrm{QFO}+QTC$ is included in $\ptime\mathrm{HBQL}$.

Next, we claim that $\ptime\mathrm{HBQL} \subseteq \mathrm{QFO}+QTC$. Note that any promise problem in $\ptime\mathrm{HBQL}$ can be assumed to have the form $(\exists \qubit{\phi_1}\in\Phi_{2^{m_1}}) (\forall  \qubit{\phi_2}\in\Phi_{2^{m_2}}) \cdots (Q_k\qubit{\phi_k}\in\Phi_{2^{m_k}}) \prob_{M}[M(\qubit{x},\Psi)= L(x)]\geq 1-\varepsilon$ for an appropriate logspace QTM $M$ running in polynomial time, where $\Psi=(\qubit{\phi_1},\qubit{\phi_2}, \ldots, \qubit{\phi_k})$ and $Q_k=\forall$ (resp., $\exists$) if $k$ is even (resp., odd).
Consider the promise problem $\LL= (L^{(+)},L^{(-)})$, where
$L^{(+)} = \{(\qubit{x},\Psi) \mid \prob_{M}[M(\qubit{x},\Psi)= 1]\geq 1-\varepsilon\}$ and $L^{(-)} = \{(\qubit{x},\Psi) \mid \prob_{M}[M(\qubit{x},\Psi) = 0]\geq 1-\varepsilon\}$.
It then follows that $\LL\in\ptime\bql$. By Theorem \ref{QFO+QTC-vs-BQL}, $\LL$ can be syntactically expressed by an appropriate formula $\phi$. By quantifying $\phi$ by $(\exists^Qy_1,|y_1|=m_1) (\forall^Qy_2,|y_2|=m_2) \cdots (Q_k y_k,|y_k|=m_k)$, we obtain the desired sentence.
\end{proof}

Since $\mathrm{FO}=\mathrm{HLOGTIME} \subseteq \dl$ \cite{BIS90} and $\dl\subseteq \ptime\bql$ (Lemma \ref{class-relation}), Theorem \ref{QFO+QTC-vs-BQL} further leads to another corollary shown below.

\begin{corollary}
$\mathrm{FO}\subseteq \mathrm{classicQFO}+QTC$.
\end{corollary}

\vs{-2}
\section{Functional Quantum Variables and Their Quantification}\label{sec:functional}

We have discussed the expressing power of $\mathrm{QFO}$ and its variant, $\mathrm{classicQFO}$, in Section \ref{sec:character-QFO} with/without quantum transitive closure (QTC) operators. The use of such operators, in fact, have helped us characterize various low-complexity classes. In this section, on the contrary, we introduce another notion of ``functional quantum variables''.

To describe a series of $O(\log{n})$ quantum operations, we need a series of $O(\log{n})$ consequential quantum variables. However, in our logical framework without QTCs, we cannot prepare such a large number of variables at once, and thus we cannot express any series of $O(\log{n})$ quantum operations. An easy solution to this difficulty is to expand the use of multiple consequential quantum variables by simply introducing a new type of variables. These new  variables are called \emph{functional quantum variables}, denoted by $Y$, $Z$, $W$, $\ldots$, which indicate functions mapping $[\ilog(n)]$ to $\Phi_{2^{\ell(n)}}$ with functions $\ell(n)$ of the form: $c$, $c\ilog(n)$, and $c\iloglog(n)$ for absolute constants $c\in\nat^{+}$.
Since a functional quantum variable $Y$ represents a ``function'',
the notation
$Y(i)$ is used to express the output qustring  $\qubit{\phi_{i}}$ of $Y$ on input $i$ and is treated as a pure-state variable. In this work, we wish to restrict the usage of functional quantum variables by demanding that  $Y(i)$ should be \emph{predecessor-dependent} for all choices of  $i$.

There is a clear distinction between two notations $Y(i)$ and $y[j]$. The former notation is used to indicate a qustring $\qubit{\psi}$ in $\Phi_{2^{\ell(n)}}$ corresponding to $i$ in $[n]$. In  the latter notation, if $y$ is a variable indicating a qustring $\qubit{\psi}$ in $\Phi_{2^{\ilog(n)}}$, then $y[j]$ expresses the $j$th qubit of the qustring $\qubit{\psi}$. It is possible to combine these two usages, such as $Y(i)[j]$.

As for the well-formedness condition, since we treat $Y(i)$ for each  $i$ as a different pure-state quantum variable, we require a similar condition imposed on quantum formulas in Section \ref{sec:syntax}.
More precisely, we first expand the notion of variable connection graph by including  all variables $Y(i)$ for any permitted number $i$. We then demand that such a variable connection graph must be \emph{topologically ordered}. A simple example of permitted expression is $(\forall i\leq \ilog(n))[R(Y(i),Y(i+1))]$.

We wish to use the term ``$\exists^Q$-functional'' to indicate the use of functional quantum variables and second-order consequential existential quantifiers that quantify these newly introduced variables.


Barrington \etalc~\cite{BIS90} proved that any languages computable by  logtime DTMs are expressible by classical first-order sentences.
An analogous statement can be proven for quantum computation.
The enhancement of the expressing power by an introduction of functional quantum variables can be demonstrated in the following statement:
any decision problem computable by bounded-error logtime quantum computations can be expressed by quantum first-order sentences together with $\exists^Q$-functional.

\begin{theorem}\label{functional-character}
$\mathrm{classicQFO}+\exists^Q$-functional $= \bqlogtime$.
\end{theorem}


For convenience, we split Theorem \ref{functional-character} into two independent lemmas, Lemmas \ref{functional-logtime} and \ref{LOGTIME-to-QFO}, and prove them separately.

\begin{lemma}\label{functional-logtime}
$\mathrm{classicQFO}+\exists^Q$-functional $\subseteq \bqlogtime$.
\end{lemma}

\begin{proof}
By Proposition \ref{classicQFO-logtime}, $\mathrm{classicQFO}\subseteq \bqlogtime$ follows. It suffices to verify that the use of functional quantum variables and their quantifications can be simulated by appropriate logtime QTMs.
Since $\exists^Q$-functional uses only quantum existential quantifiers over functional quantum variables, we intend to demonstrate that an appropriate logtime QTM can simulate quantum existential quantifiers over functional quantum variables.

In what follows, we assume that a quantum sentence $\phi$ has the form $(\exists^Q Y) \psi(Y)$, where $Y$ is a functional quantum variable representing  a function mapping $[\ilog(n)]$ to $\Psi_{2^{\ell(n)}}$ and $\psi$ is a quantum formula containing no second-order consequential existential quantifier.
Assuming that $\psi$ is simulatable in log time, we wish to simulate $\phi$ on an appropriate logtime QTM, say, $M$ as follows. We evaluate  $Y(i)$ as a quantum state of $\ilog(n)$ qubits and keep this quantum state on an $\ilog(n)$-qubit work tape of $M$. An actual simulation process is sone similarly to the proof of Proposition  \ref{classicQFO-logtime} by treating $Y(i)$ as a pure-state quantum variable.
To differentiate those variables $Y(i)$, moreover, we need to keep track of index $i$. This can be done in log time because $i$ is taken from $[\ilog(n)]$.
\end{proof}


\begin{lemma}\label{LOGTIME-to-QFO}
$\mathrm{BQLOGTIME}\subseteq \mathrm{classicQFO}+\exists^Q$-functional.
\end{lemma}

\begin{proof}
Let us take any promise problem $\LL=(L^{(+)},L^{(-)})$ and any logtime QTM $M$ that solves it with probability at least $1-\varepsilon$, where $\varepsilon\in[0,1/2)$. Take a constant $c\in\nat^{+}$ for which $M$ halts within $c\ilog(n)$ steps.
Since the runtime bound of $M$ is $c\ilog(n)$, we split the integer interval $[c\ilog(n)]$ into $c$ blocks of equal size $\ilog(n)$: $[k\ilog(n)+1,(k+1)\ilog(n)]_{\integer}$ for each number $k\in[0,c-1]_{\integer}$.
We simulate the behavior of $M$ during the time period specified by each of these integer intervals.
Let $t_k=k\ilog(n)+t$ for $k\in[0,c-1]_{\integer}$ and $t\in[0,\ilog(n)-1]_{\integer}$.

Since $M$'s input tape is read-only, it suffices to consider $M$'s surface  configurations $(q,l_1,l_2,w,l_3,z)$, where $q$ is the current inner state, $l_1$, $l_2$, and $l_3$ are respectively the current tape head locations of an input tape, an index tape, and a work tape, and $w$ and $z$ are respectively the current contents of the index tape and the work tape.
It suffices to simulate the changes of the content of the index tape of $M$ step by step.

Let us explain how to handle an index tape and a work tape of $M$. We begin with the index tape. Note that the index tape uses only $\ilog(n)$ qubits together with two designated endmarkers.  Let $\ell=2\ilog(n)$.

Let us recall the function $STEP$ from the proof of Proposition \ref{QFO+logQTC-bounds}, which represents a single step of $M$ on pinpoint configurations. The function $STEP$ maps $\qubit{q,x_{(l_1)},w_{(l_2)},z_{(l_3)},l_1,l_2,l_3}$ to $\cos\theta\qubit{p,x_{(l_1)},\tau,\xi,l_1,l_2+d_1,l_3+d_3} + \sin\theta \qubit{p',x_{(l_1)},\tau',\xi',l_1,l_2+d'_1,l_3+d'_3}$ for a constant $\theta\in[0,2\pi)$.
Here, we intend to define a quantum formula, say, $COMP_{t_k}$, which expresses $STEP$ using functional quantum variables.

We first concentrate on the index tape.
We introducing two groups of $c$ functional quantum variables, $Y_1,Y_2,\ldots,Y_c$ and $W_1,W_2,\ldots,W_c$, where each of $Y_k$ and $W_k$ maps $[\ilog(n)]$ to $\Phi_{2^{\ell(n)}}$ and to $\Phi_{2^{\iloglog(n)}}$, corresponding to the $k$th block $[k\ilog(n)+1,(k+1)\ilog(n)]_{\integer}$.   Let $YW_i(t)$ denote $Y_i(t)\otimes W_i(t)$ for convenience.
For the content of cell $i$ of the index tape together with its tape head location $l_2$ at time $t\in T_k$, $W_k(t)$ represents $l_2$ in unary  as $0^{l_2-1}10^{c\ilog(n)-l_2}$, and $Y_k(t)[2i-1,2i]$ represents the encoding of $w_{(l_2)}$, where the tape symbols $0$, $1$, and $B$ are encoded into $\hat{0}$, $\hat{1}$, and $\hat{B}$.
Here, we use $Y_k(t)[2i-1,2i] \otimes W_k(t)$, where $i,t\in[\ilog(n)]$.

Following to the proof of Proposition \ref{QFO+logQTC-bounds}, we write $sym_i$ and $hp_i$ to indicate the content of the $i$th tape cell and the presence of the tape head at cell $i$ at time $t_k$, respectively. Remark that they are induced from $Y_k(t)[2i-1,2i]\otimes W_k(t)$. Similarly, $newsym_i$ and $newhp_i$ are induced from $Y_k(t+1)[2i-1,2i]\otimes W_k(t+1)$.

Assuming that $(hp_{i-1},hp_i,hp_{i+1})$ takes the value of $(1,0,0)$, if $\hat{d}_1=0$, then we set $Q_{(3)}$ to be $\PP_{copy}(sym_i:newsym_i)\wedge \PP_{copy}(1:newhp_i)$; otherwise, we set $Q_{(3)}$ to be $\PP_{copy}(sym_i:newsym_i)\wedge \PP_{copy}(0:newhp_i)$. In the case where $(hp_{i-1},hp_i,hp_{i+1})$ takes $(0,1,0)$, we set $Q_{(4)}\equiv \PP_{copy}(\hat{\tau}:newsym_i)\wedge \PP_{copy}(0:newhp_i)$.
If $(hp_{i-1},hp_i,hp_{i+1})$ takes $(0,0,0)$, then we set $Q_{(1)}\equiv \PP_{copy}(sym_i:newsym_i)\wedge \PP_{copy}(hp_i:newhp_i)$. When $(hp_{i-1},hp_i,hp_{i+1})$ takes  $(0,0,1)$, if $\hat{d}_1=0$, then $\PP_{copy}(sym_i:newsym_i)\wedge \PP_{copy}(0:newhp_i)$. If $\hat{d}_1=1$, then $\PP_{copy}(sym_i:newsym_i)\wedge \PP_{copy}(1:newhp_i)$.

In a similar way, we can handle the work tape of $M$. Note that the work tape uses at most $c\ilog(n)$ tape cells. We then introduce functional quantum variables $Z_1,Z_2,\ldots,Z_c$ and $V_1,V_2,\ldots,V_c$, where each of $Z_k$ and $V_k$ maps $[\ilog(n)]$ to $\Phi_{2^{m(n)}}$ and to $\Phi_{2^{\iloglog(n)}}$, where $m(n)=c\ilog(n)$.

An entire computation of $M$ is described by $c\ilog(n)$ applications of $STEP$. Hence, $YW_1,YW_2,\ldots,YW_{c\ilog(n)}$ logically express this entire computation.

When $M$ finally halts, we then observe the designated cell to determine the outcome of $M$'s computation. To express this, we use the quantum term of the form $s\simeq_{\varepsilon}1$.
\end{proof}

From Theorem \ref{functional-character} follows the following corollary.

\begin{corollary}\label{QFO-vs-HQC}
$\mathrm{QFO}+\exists^Q$-functional $=  \mathrm{HQBLOGTIME}$.
\end{corollary}

\begin{proof}
Firstly, we wish to show that $\mathrm{HBQLOGTIME} \subseteq \mathrm{QFO}+\exists^Q$-functional. For any promise problem $\LL= (L^{(+)},L^{(-)})$ in $\mathrm{HBQLOGTIME}$, we take a logtime QTM $M$ such that (i) for any $\qubit{\phi}\in L^{(+)}$, $(\exists\qubit{\psi_1}\in\Phi_{2^m})(\forall \qubit{\psi_2}\in\Phi_{2^m}) \cdots (Q_k\qubit{\psi_k}\in\Phi_{2^m}) [ \prob_{M}[M(\qubit{\phi},\Psi)=1]\geq 1-\varepsilon]$ and (ii) for any $\qubit{\phi}\in L^{(-)}$, $(\forall\qubit{\psi_1}\in\Phi_{2^m})(\exists \qubit{\psi_2}\in\Phi_{2^m}) \cdots (\bar{Q}_k\qubit{\psi_k}\in\Phi_{2^m}) [ \prob_{M}[M(\qubit{\phi},\Psi)=0]\geq 1-\varepsilon]$, where $\Psi=(\qubit{\psi_1},\qubit{\psi_2}, \ldots, \qubit{\psi_k})$ and $Q_k=\forall$ if $k$ is odd and $\exists$ otherwise.

We then define $K^{(+)} = \{(\qubit{\phi},\Xi)\mid \prob[M(\qubit{\phi},\Xi)=1]\geq 1-\varepsilon\}$ and $K^{(-)} = \{(\qubit{\phi},\Xi)\mid \prob[M(\qubit{\phi},\Xi)=0]\geq 1-\varepsilon\}$, where $\Xi = (\qubit{\psi_1},\qubit{\psi_2}, \ldots, \qubit{\psi_k})$.
We claim that $\KK=(K^{(+)},K^{(-)})$ belongs to $\mathrm{BQLOGTIME}$. Because of $\KK\in \mathrm{BQLOGTIME}$, by Lemma \ref{LOGTIME-to-QFO}, we obtain a quantum formula $\zeta$ syntactically expressing $\KK$ within $\mathrm{claasicQFO}+\exists^Q$-functional.
By applying quantum quantifiers, we conclude that $\LL$ belongs to  $\mathrm{QFO}+\exists^Q$-functional.

Next, we show the opposite direction: $\mathrm{QFO}+\exists^Q$-functional $\subseteq \mathrm{HBQLOGTIME}$.
Let $\LL= (L^{(+)},L^{(-)})$ be any problem in $\mathrm{QFO}+\exists^Q$-functional  and take a quantum sentence $\phi$ that syntactically  expresses $\LL$. This $\phi$ can be assumed to have the form $(\exists^QY_1)(\forall^QY_2)\cdots (Q_kY_k)\psi(X,Y_1,Y_2,\ldots,Y_k)$, where $\psi$ is second-order quantifier-free and $Q_j=\forall$ if $j$ is even and $\exists$ otherwise.
We thus want to show how to simulate $\exists^Q$-functional.
By a simple calculation conducted in the proof of Proposition \ref{classicQFO-logtime}, we can simulate $\psi$ on an appropriate logtime QTM $M$. By translating second-order quantum quantifiers $(Q_jY_j)$ for $Q_j\in\{\forall^Q,\exists^Q\}$ into $(Q_j\qubit{\psi_j}\in\Phi_{2^m})$, it follows that $\LL\in\mathrm{HBQLOGTIME}$.
\end{proof}

\vs{-2}
\section{Brief Discussion and Future Work}

It has been well-known that \emph{logical expressibility} is usable as a meaningful complexity measure to assert the ``difficulty'' of solving  various combinatorial problems in sharp contrast with \emph{machine-based computability}.
Nonetheless, there are known connections between expressibility and computability in the classical setting.
In this work, we have introduced a quantum analogue of first-order logic to logically express quantum computations in a good hope for a better understanding of quantum computing. We have made a partial success in setting up a reasonably good framework to discuss the intended quantum first-order logic whose formulation is founded on (recursion-theoretic) schematic definitions of quantum functions of \cite{Yam20,Yam24}.  The use of these schemes to capture quantum polynomial-time and polylogarithmic-time computing has made a pavement to our QFO formalism.
Throughout Sections \ref{sec:character-QFO}--\ref{sec:functional}, $\mathrm{QFO}$ and its variant $\mathrm{classicQFO}$ have been proven to be good bases to characterize logarithmic-time/space quantum computing.

To promote the deeper understandings of quantum logic, we wish to list a few interesting questions for the sake of the avid reader.

\renewcommand{\labelitemi}{$\circ$}
\begin{enumerate}
  \setlength{\topsep}{-2mm}%
  \setlength{\itemsep}{1mm}%
  \setlength{\parskip}{0cm}%

\item The usefulness of $\mathrm{QFO}$ as well as $\mathrm{classicQFO}$ has been demonstrated by precisely characterizing several low complexity classes in Sections \ref{sec:character-QFO}--\ref{sec:functional} with an excessive use of $logQTC$, $QTC$, and functional quantum variables. It is imperative to find more quantum complexity classes that can be logically expressed based on $\mathrm{QFO}$ and $\mathrm{classicQFO}$.

\item Unfortunately, we still do not know the precise complexity of $\mathrm{QFO}$ and $\mathrm{classicQFO}$ without $logQTC$, $QTC$, or  functional quantum variables. Therefore, one important remaining task is to determine the precise complexity of them.

\item We have distinguished between consequential and introductory quantum quantifications. One possible way to eliminate any use of consequential quantum existential quantifiers in $\mathrm{QFO}$ is an introduction of  a \emph{quantum $\mu$-operator} in the form of $\mu z.\phi(x,z)$ for a quantum formula $\phi$. It is interesting to explore the properties of such an operator.

\item Classical proof complexity theory has been developed over the years. See a textbook, e.g., \cite{CN10} for references. We hope to introduce a ``natural'' deduction system and a ``natural'' sequential calculus for the quantum first-order logics and develop quantum proof complexity theory.

\item As a direct extension of $\mathrm{QFO}$,  we have discussed in Section \ref{sec:functional} second-order quantum quantifications to handle functional quantum variables. By the use of second-order variables, $X(y)$ expresses the quantum state obtained by applying a second-order variable $X$ to a first-order variable $y$. For example, a quantum formula of the form $\exists^Q X\forall y\forall z [ \PP_{CNOT}(y:z) \Leftrightarrow I(X(y):z)]$ makes it possible to use $X$ in place of $CNOT$.
    It is interesting to develop full-fledged quantum second-order logics that can characterize a wider rage of complexity classes.
\end{enumerate}



\let\oldbibliography\thebibliography
\renewcommand{\thebibliography}[1]{%
  \oldbibliography{#1}%
  \setlength{\itemsep}{-2pt}%
}

\bibliographystyle{alpha}

\end{document}